\documentclass[12pt]{article}

\usepackage{amsmath}
\usepackage{amssymb}
\usepackage{amsthm}
\usepackage{bm}
\usepackage[footnotesize,bf]{caption}
\usepackage{color}
\usepackage[outline]{contour}
\usepackage{eucal}
\usepackage{fullpage}
\usepackage{graphicx}
\usepackage{hhline}
\usepackage{lineno}
\usepackage{longtable}
\usepackage[sc]{mathpazo}
\usepackage{mathrsfs}
\usepackage{mathtools}
\usepackage[round]{natbib}
\usepackage{setspace}
\usepackage{times}
\usepackage{xr}
\usepackage[all,cmtip]{xy}
\usepackage{verbatim}

\newtheorem{corollary}{Corollary}
\newtheorem{lemma}{Lemma}
\newtheorem{proposition}{Proposition}
\newtheorem{theorem}{Theorem}

\theoremstyle{definition}
\newtheorem{definition}{Definition}
\newtheorem{example}{Example}
\newtheorem{remark}{Remark}

\newtheoremstyle{runningex}{3pt}{3pt}{}{}{\itshape}{.}{0.5em}{}
\theoremstyle{runningex}
\newtheorem*{moranex}{Example: Moran process}

\newcommand*\patchAmsMathEnvironmentForLineno[1]{%
	\expandafter\let\csname old#1\expandafter\endcsname\csname #1\endcsname
	\expandafter\let\csname oldend#1\expandafter\endcsname\csname end#1\endcsname
	\renewenvironment{#1}%
	{\linenomath\csname old#1\endcsname}%
	{\csname oldend#1\endcsname\endlinenomath}}%
\newcommand*\patchBothAmsMathEnvironmentsForLineno[1]{%
	\patchAmsMathEnvironmentForLineno{#1}%
	\patchAmsMathEnvironmentForLineno{#1*}}%
\AtBeginDocument{%
	\patchBothAmsMathEnvironmentsForLineno{equation}%
	\patchBothAmsMathEnvironmentsForLineno{align}%
	\patchBothAmsMathEnvironmentsForLineno{flalign}%
	\patchBothAmsMathEnvironmentsForLineno{alignat}%
	\patchBothAmsMathEnvironmentsForLineno{gather}%
	\patchBothAmsMathEnvironmentsForLineno{multline}%
}

\DeclareMathOperator{\Prob}{\mathbb{P}}
\DeclareMathOperator{\E}{\mathbb{E}}
\newcommand{\delsel}{\Delta_{\mathrm{sel}}}

\newcommand{\delhat}{\widehat{\Delta}}
\newcommand{\delhatsel}{\widehat{\Delta}_{\mathrm{sel}}}

\newcommand{\vx}{\mathbf{x}}
\newcommand{\vy}{\mathbf{y}}

\newcommand{\vA}{\mathbf{A}}
\newcommand{\vB}{\mathbf{B}}

\newcommand{\vF}{\mathbf{F}}

\newcommand{\bxi}{\bm{\xi}}
\newcommand{\bB}{\mathbb{B}}

\newcommand{\MSS}{\circlearrowright}

\newcommand{\T}{\intercal}
\newcommand{\eq}[1]{Eq.~(\ref{eq:#1})}
\newcommand{\eqsa}[2]{Eqs.~(\ref{eq:#1})~and~(\ref{eq:#2})}

\newcommand{\eqsc}[2]{Eqs.~(\ref{eq:#1})--(\ref{eq:#2})}
\newcommand{\fig}[1]{Fig.~\ref{fig:#1}}
\newcommand{\cor}[1]{Corollary~\ref{cor:#1}}
\newcommand{\defn}[1]{Definition~\ref{def:#1}}
\newcommand{\ex}[1]{Example~\ref{ex:#1}}
\newcommand{\lem}[1]{Lemma~\ref{lem:#1}}
\newcommand{\prop}[1]{Proposition~\ref{prop:#1}}
\newcommand{\sect}[1]{Section~\ref{sec:#1}}
\newcommand{\tthm}[1]{Theorem~\ref{thm:#1}}
\newcommand{\tab}[1]{Table~\ref{tab:#1}}
\newtheorem*{fixation}{Fixation Axiom}
\newcommand{\fix}{Fixation Axiom}

\title{\begin{center} \bfseries \singlespacing
		Fixation probabilities in evolutionary dynamics under weak selection
\end{center}}
\author{\parbox[c]{16cm}{\onehalfspacing \normalsize \centering ~\\[-0.4cm] Alex McAvoy$^{1,2,3,}$\footnote{ORCID: \texttt{https://orcid.org/0000-0002-9110-4635}} \quad Benjamin Allen$^{4,}$\footnote{ORCID: \texttt{https://orcid.org/0000-0002-9746-7613}} \\ \quad \\
    \small
		$^{1}$Department of Organismic and Evolutionary Biology, Harvard University, Cambridge, MA~02138 \\
		$^{2}$Department of Mathematics, University of Pennsylvania, Philadelphia, PA~19104 \\
		$^{3}$Center for Mathematical Biology, University of Pennsylvania, Philadelphia, PA~19104 \\
		\texttt{amcavoy@sas.upenn.edu} \\ \quad \\ 
		$^{4}$Department of Mathematics, Emmanuel College, Boston, MA~02115 \\
		\texttt{allenb@emmanuel.edu}\\[0.2cm]}
	\date{}
}

\begin{document}
	
\allowdisplaybreaks

\maketitle

\begin{abstract}
In evolutionary dynamics, a key measure of a mutant trait's success is the probability that it takes over the population given some initial mutant-appearance distribution. This ``fixation probability'' is difficult to compute in general, as it depends on the mutation's effect on the organism as well as the population's spatial structure, mating patterns, and other factors. In this study, we consider weak selection, which means that the mutation's effect on the organism is small. We obtain a weak-selection perturbation expansion of a mutant's fixation probability, from an arbitrary initial configuration of mutant and resident types. Our results apply to a broad class of stochastic evolutionary models, in which the size and spatial structure are arbitrary (but fixed). The problem of whether selection favors a given trait is thereby reduced from exponential to polynomial complexity in the population size, when selection is weak. We conclude by applying these methods to obtain new results for evolutionary dynamics on graphs.
\end{abstract}

\section{Introduction}
Many studies of stochastic evolutionary dynamics concern the competition of two types (traits) in a finite population. Through a series of births and deaths, the composition of the population changes over time. Absent recurring mutation, one of the two types will eventually become fixed and the other will go extinct. Such models may incorporate frequency-dependent selection as well as spatial structure or other forms of population structure. The fundamental question is to identify the selective differences that favor the fixation of one type over the other.

This question is typically answered by computing a trait's fixation probability \citep{haldane1927mathematical,moran:MPCPS:1958,kimura:G:1962,patwa:JRSI:2008,traulsen:W:2010,der:TPB:2011,hindersin:JRSI:2014,mccandlish:TPB:2015,hindersin:B:2016} as a function of the initial configuration of the two types. Direct calculation of a mutant's fixation probability is possible in simple models of well-mixed populations \citep{moran:MPCPS:1958,taylor:BMB:2004} or spatially structured populations that are highly symmetric \citep{lieberman:Nature:2005,broom:PRSA:2008,hindersin:JRSI:2014} or small \citep{hindersin:PLOSCB:2015,cuesta2018evolutionary,moller2019exploring,tkadlec2019population}. For large populations, fixation probabilities can sometimes be approximated using diffusion methods \citep{kimura:G:1962,roze2003selection,ewens:S:2004,chen2018wright}.

When selection is weak, perturbative methods can be applied to the computation of fixation probabilities \citep{haldane1927mathematical,nowak:Nature:2004,lessard:JMB:2007}. The first-order effects of selection on fixation probabilities provide information about whether selection favors one trait over another, and, if so, by how much. This perturbative approach is often paired with methods from coalescent theory \citep{rousset:JTB:2003,chen:AAP:2013,van2015social,chen:SR:2016,allen:Nature:2017,allen2020transient}.

Our aim in this work is to generalize the weak-selection method for computing fixation probability to a broad class of evolutionary models. Our main result is a first-order weak-selection expansion of a mutant's fixation probability from any initial condition. This result applies to arbitrary forms of spatial structure and frequency-dependent selection, and the expansion can be computed for any particular model and initial configuration by solving a system of linear equations. Under conditions that apply to most models of interest, the size of this system---and hence the  complexity of computing this expansion---exhibits polynomial growth in the population size.

Our approach is based on a modeling framework developed by \cite{allen:JMB:2014} and \cite{allen:JMB:2019}, which is described in \sect{model}. This framework describes stochastic trait transmission in a population of fixed size, $N$, and spatial structure. This setup leads to a finite Markov chain model of selection. Special cases of this framework include the Moran \citep{moran:MPCPS:1958} and Wright-Fisher models \citep{fisher:OUP:1930,wright:G:1931,imhof:JMB:2006}, as well as evolutionary games in graph-structured populations \citep{ohtsuki:Nature:2006,szabo:PR:2007,nowak:PTRSB:2009,allen:Nature:2017}. We use this framework to define the \emph{degree} of an evolutionary process, which later plays an important role in determining the computational complexity of calculating fixation probabilities.

In \sect{sojourn}, we establish a connection between sojourn times and stationary probabilities. Specifically, we compare the original Markov chain, which is absorbing, to an amended Markov chain, in which the initial configuration $\bxi$ can be re-established from the all-$A$ and all-$B$ states with some probability, $u$. This amended Markov chain is recurrent, and it has a unique stationary distribution. We show that sojourn times for transient states of the original Markov chain are equal to $u$-derivatives, at $u=0$, of stationary probabilities in the amended chain. We also define a set-valued coalescent process that is used in the proof of our main results.

\sect{fixation} proves our main result regarding fixation probabilities. We consider the fixation probability, $\rho_{A}\left(\bxi\right)$, of type $A$ in a population whose initial configuration of $A$ and $B$ is $\bxi$. The intensity of selection, $\delta$, which quantifies selective differences between the two types, is assumed to be sufficiently weak, meaning $\delta\ll 1$. Our main result is a formula for calculating $\rho_{A}\left(\bxi\right)$ to first order in the selection intensity $\delta$. This formula depends on the first-order effects of selection on marginal trait-transmission probabilities together with a set of sojourn times for neutral drift. The latter can be evaluated by solving a linear system of size $O\left(N^{D+1}\right)$, where $D$ is the degree of the process. This linear system is what bounds the complexity in $N$ of calculating the first-order coefficient, $\frac{d}{d\delta}\Big\vert_{\delta =0}\rho_{A}\left(\bxi\right)$.

In \sect{mutappear}, we extend our main result to the case that the initial configuration of $A$ and $B$ is stochastic rather than deterministic. We derive a formula for $\frac{d}{d\delta}\Big\vert_{\delta =0}\E_{\mu_{A}}\left[\rho_{A}\right]$, where $\mu_{A}$ is an arbitrary distribution over initial configurations of $A$ and $B$ (and which can also depend on $\delta$). \sect{comparing} considers relative measures of evolutionary success \citep{tarnita:AN:2014} obtained by comparing $\E_{\mu_{A}}\left[\rho_{A}\right]$ to $\E_{\mu_{B}}\left[\rho_{B}\right]$ when $A$ and $B$ each have their own initial distributions, $\mu_{A}$ and $\mu_{B}$.

Finally, we apply our results to several well-known questions in evolutionary dynamics, particularly on graph-structured populations. \sect{constant} discusses the case of constant fecundity, wherein the reproductive rate of an individual depends on only its own type. A large body of research \citep{lieberman:Nature:2005,broom:PRSA:2008,broom:JSTP:2011,voorhees:PRSA:2013,monk:PRSA:2014,hindersin:PLOSCB:2015,kaveh:RSOS:2015,cuesta2017suppressors,pavlogiannis:CB:2018,moller2019exploring,tkadlec2019population} aims to understand the effects of graph structure on fixation probabilities in this context. Our results provide efficient recipes to calculate fixation probabilities under weak selection. \sect{examples} turns to evolutionary game theory. For a particular prisoner's dilemma game (the ``donation game''; \citealp{sigmund:PUP:2010}), we derive a formula for the fixation probability of cooperation from any starting configuration on an arbitrary weighted graph, generalizing and unifying a number of earlier results \citep{ohtsuki:Nature:2006,taylor:Nature:2007,chen:AAP:2013,allen:EMS:2014,chen:SR:2016,allen:Nature:2017}.

\section{Modeling evolutionary dynamics}\label{sec:model}
We employ a framework previously developed by \cite{allen:JMB:2014} and \cite{allen:JMB:2019} to represent an evolving population with arbitrary forms of spatial structure and frequency dependence. In addition to being described below, all of the notation and symbols we use are outlined in \tab{notationtable}. Although we will use the language of a haploid asexual population, our formalism applies equally well to diploid, haplodiploid, or other populations by considering the alleles to be asexual replicators (the ``gene's-eye view''), as described in \cite{allen:JMB:2019}.

As a source of motivation and a tool to illustrate the general model, we consider the Moran process \citep{moran:MPCPS:1958} as a running example. The Moran process models evolution in an unstructured population consisting of two types, a mutant ($A$) and a resident ($B$), of relative fecundity (reproductive rate) $r$ and $1$, respectively. In a population consisting of $i$ individuals of type $A$ and $N-i$ individuals of type $B$, the probability that type $A$ is selected to reproduce is $ir/\left(ir+N-i\right)$. With probability $\left(N-i\right) /\left(ir+N-i\right)$, type $B$ is selected to reproduce. The offspring, which inherits the type of the parent, then replaces a random individual in the population. Throughout our discussion of the framework below, we return to this simple example repeatedly (with more sophisticated examples following in later sections).

\subsection{Modeling assumptions}\label{sec:modeling_assumptions}
We consider competition between two alleles, $A$ and $B$, in a population of finite size, $N$. The state of the population is given by $\vx\in\left\{0,1\right\}^{N}$, where $x_{i}=1$ (resp. $x_{i}=0$) indicates that individual $i$ has type $A$ (resp. $B$). 

Since we are mainly concerned with the probability of fixation rather than the timescale, we may assume without a loss of generality that the population evolves in discrete time. However, the results reported here can also be applied to continuous-time models in a straightforward manner. In what follows, we assume that the population's state is updated in discrete time steps via replacement events. A replacement event is a pair, $\left(R,\alpha\right)$, where $R\subseteq\left\{1,\dots ,N\right\}$ is the set of individuals who are replaced in a given time step and $\alpha :R\rightarrow\left\{1,\dots ,N\right\}$ is the offspring-to-parent map. For a fixed replacement event, $\left(R,\alpha\right)$, the state of the population at time $t+1$, $\vx^{t+1}$, is obtained from the state of the population at time $t$, $\vx^{t}$, by letting
\begin{align}
\label{eq:stateupdate}
x_{i}^{t+1} &= 
\begin{cases}
\displaystyle x_{\alpha\left(i\right)}^{t} & \displaystyle i\in R , \\
& \\
\displaystyle x_{i}^{t} & \displaystyle i\not\in R .
\end{cases}
\end{align}
We can express such a transition more concisely by defining the extended mapping
\begin{align}
\widetilde{\alpha} &: \left\{1,\dots ,N\right\} \longrightarrow \left\{1,\dots ,N\right\} \nonumber \\
&: i \longmapsto 
\begin{cases}
\displaystyle \alpha\left(i\right) & \displaystyle i\in R , \\
& \\
\displaystyle i & \displaystyle i\not\in R .
\end{cases}
\end{align}
We then have $\vx^{t+1} = \vx_{\widetilde{\alpha}}^{t}$, where $\vx_{\widetilde{\alpha}}$ is the vector whose $i$th component is $x_{\widetilde{\alpha}\left(i\right)}$. 

In state $\mathbf{x}$, we denote by $\left\{p_{\left(R,\alpha\right)}\left(\vx\right)\right\}_{\left(R,\alpha\right)}$ the distribution from which the replacement event is chosen. We call this distribution (as a function of $\vx$) the \emph{replacement rule}. We assume that this replacement rule depends on an underlying parameter $\delta\geqslant 0$ that represents the intensity of selection. Neutral drift corresponds to $\delta=0$, and weak selection is the regime $\delta \ll 1$.

\begin{moranex}
In the Moran process, a slightly advantageous mutant has fecundity $r=1+\delta$ for $\delta\ll 1$. The probability that replacement event $\left(R,\alpha\right)$ is chosen is
\begin{align}
p_{\left(R,\alpha\right)}\left(\vx\right) &= 
\begin{cases}
 \displaystyle \frac{x_{\alpha\left(i\right)}r+1-x_{\alpha\left(i\right)}}{\sum_{j=1}^{N}\left(x_{i}r +1-x_{i}\right)} \frac{1}{N} & \displaystyle R=\left\{i\right\} , \\
 & \\
 \displaystyle 0 & \displaystyle \left| R\right|\neq 1 .
\end{cases}\label{eq:moranReplacementRule}
\end{align}
That is, assuming $R=\left\{i\right\}$, the probability that $\alpha\left(i\right)$ reproduces is proportional to $x_{\alpha\left(i\right)}r +1-x_{\alpha\left(i\right)}$ (which is $r$ if $\alpha\left(i\right)$ has type $A$ and $1$ otherwise). The constant of proportionality is the reciprocal of the total population fecundity, $\sum_{j=1}^{N}\left(x_{i}r +1-x_{i}\right)$. If $\alpha\left(i\right)$ reproduces, then the probability that the offspring replaces $i$ is simply $1/N$.
\end{moranex}

For brevity, we write $\mathbb{B}\coloneqq\left\{0,1\right\}$. Each replacement rule defines a Markov chain on $\mathbb{B}^{N}$, according to \eq{stateupdate}. We let $P_{\vx\rightarrow\vy}$ be the probability of transitioning from state $\vx$ to state $\vy$ in this Markov chain.

We make three assumptions on the replacement rule. The first is that for every $\delta$, there exists at least one individual who can generate a lineage that takes over the entire population. We state this assumption as an axiom:
\begin{fixation}
There exists $i\in\left\{1,\dots ,N\right\}$, $m \geqslant 1$, and a sequence $\left\{\left(R_{k},\alpha_{k}\right)\right\}_{k=1}^{m}$ with
\begin{itemize}

\item $p_{\left(R_{k},\alpha_{k}\right)}\left(\vx\right) >0$ for every $k\in\left\{1,\dots ,m\right\}$ and $\vx\in\mathbb{B}^{N}$;

\item $i\in R_{k}$ for some $k\in\left\{1,\dots ,m\right\}$;

\item for every $j\in\left\{1,\dots ,N\right\}$, we have $\widetilde{\alpha}_{1}\circ\widetilde{\alpha}_{2}\circ\cdots\circ\widetilde{\alpha}_{m}\left(j\right) =i$.

\end{itemize}
\end{fixation}
The requirement that $i\in R_{k}$ for some $k\in\left\{1,\dots ,m\right\}$ guarantees that the individual at $i$ cannot live forever, since otherwise no evolution would occur \citep{allen:JMB:2019}. \cite{allen:JMB:2014} showed that, under the {\fix}, this Markov chain has two absorbing states: the state $\vA$ in which $x_{i}=1$ for every $i$ (all-$A$), and the state $\vB$ in which $x_{i}=0$ for every $i$ (all-$B$). All other states are transient. We denote the set of all transient states by $\mathbb{B}_{\T}^{N}\coloneqq\mathbb{B}^{N}-\left\{\vA ,\vB\right\}$.

Our second assumption is that when $\delta =0$ (neutral drift), the replacement rule does not depend on the state, $\vx$. In this case, we denote the replacement rule by $\left\{p_{\left(R,\alpha\right)}^{\circ}\right\}_{\left(R,\alpha\right)}$. Note that we have removed the dependence on $\vx$. This assumption arises because, under neutral drift, the competing alleles are interchangeable, and so the probabilities of replacement should not depend on how these alleles are distributed among individuals. More generally, in the quantities derived from the replacement rule below (e.g. birth rates and death probabilities), we use the superscript $\circ$ to denote their values under neutral drift.

Our third assumption is that for every $\vx\in\mathbb{B}^{N}$ and every replacement event $\left(R,\alpha\right)$, $p_{\left(R,\alpha\right)}\left(\vx\right)$ is a smooth function of $\delta$ in a small neighborhood of $\delta =0$. This assumption enables a perturbation expansion in the selection strength, $\delta$.

\begin{moranex}
Smoothness in $\delta$ is evident from \eq{moranReplacementRule} whenever $r$ is itself a smooth function of $\delta$. 
\end{moranex}

\subsection{Quantifying selection}\label{sec:quantifying_selection}
Having outlined the class of models under consideration, we now define quantities that characterize natural selection in a given population state $\vx \in \bB^N$. For any $i$ and $j$ in $\left\{1,\dots ,N\right\}$, let $e_{ij}\left(\vx\right)$ be the marginal probability that $i$ transmits its offspring to $j$ in state $\vx$, i.e.
\begin{align}
e_{ij}\left(\vx\right) &\coloneqq \sum_{\substack{\left(R,\alpha\right) \\ j\in R,\,\alpha\left(j\right) =i}} p_{\left(R,\alpha\right)}\left(\vx\right) . \label{eq:marginal}
\end{align}
The birth rate (expected offspring number) of $i$ is $b_{i}\left(\vx\right)\coloneqq\sum_{j=1}^{N}e_{ij}\left(\vx\right)$ and the death probability of $i$ is $d_{i}\left(\vx\right)\coloneqq\sum_{j=1}^{N}e_{ji}\left(\vx\right)$. Using these quantities, we can write the expected change in the frequency of $A$ due to selection as \citep{tarnita:AN:2014,allen:JMB:2019}
\begin{align}
\label{eq:delsel}
\delsel\left(\vx\right) &\coloneqq \sum_{i=1}^{N}x_{i}\left(b_{i}\left(\vx\right) -d_{i}\left(\vx\right)\right) .
\end{align}

Any real-valued function on $\bB^N$ is called a \emph{pseudo-Boolean function} \citep{hammer:S:1968}. Since $e_{ij}\left(\vx\right)$ and its derivative with respect to $\delta$ at $\delta =0$ are pseudo-Boolean functions, for every $i$ and $j$ there is a unique multi-linear polynomial representation \citep{hammer:S:1968,boros:DAM:2002},
\begin{align}
\frac{d}{d\delta}\Bigg\vert_{\delta =0}e_{ij}\left(\vx\right) &= \sum_{I\subseteq\left\{1,\dots ,N\right\}} c_{I}^{ij} \vx_{I} , \label{eq:eij_derivative}
\end{align}
where the $c_{I}^{ij}$ are a collection of real numbers (Fourier coefficients) indexed by the subsets $I \subseteq \left\{1,\dots ,N\right\}$, and $\vx_{I}\coloneqq\prod_{i\in I}x_{i}$. Note that $\vx_{I}$ is a scalar, not a state, with $\vx_{I}=1$ if and only if $x_i=1$ for each $i \in I$. This representation includes the constant term $c_\varnothing^{ij}$, linear terms of the form  $c_{\left\{k\right\}}^{ij}x_k$, quadratic terms of the form $c_{\left\{h, k\right\}}^{ij}x_hx_k$, and so on up through $c_{\left\{1,\ldots,N\right\}}^{ij}x_1 \cdots x_N$. The coefficients $c^{ij}_I$ quantify how genetic assortment among sets of individuals affects the probability that $i$'s offspring replaces $j$, under weak selection.

We let $D_{ij}$ denote the degree of the above representation, defined as the degree of the highest-order nonzero term:
\begin{align}
D_{ij} &\coloneqq \max\left\{ k\ :\ c_{I}^{ij}\neq 0 \textrm{ for some } I\subseteq\left\{1,\dots ,N\right\} \textrm{ with } \left| I\right| = k \right\}.
\end{align}
(In the trivial case that $c_{I}^{ij}=0$ for every $I\subseteq\left\{1,\dots ,N\right\}$, we set $D_{ij}=0$.)

For $I\subseteq\left\{1,\dots ,N\right\}$, let $\mathbf{1}_{I}\in\mathbb{B}^{N}$ denote the state in which $x_{i}=1$ for $i\in I$ and $x_{i}=0$ for $i\not\in I$. By applying a M\"{o}bius transform to \eq{eij_derivative} \citep{grabisch:MOR:2000}, we can express the coefficients $c_{I}^{ij}$ as
\begin{align}
c_{I}^{ij} &= \frac{d}{d\delta}\Bigg\vert_{\delta =0} \sum_{J\subseteq I} \left(-1\right)^{\left| I\right| -\left| J\right|} e_{ij}\left(\mathbf{1}_{J}\right) . \label{eq:mobius}
\end{align}
This expression provides a recipe for calculating the coefficients $c_{I}^{ij}$ directly from $e_{ij}\left(\vx\right)$ for a given process (see \sect{examples}).

We define the degree of the overall evolutionary process, under weak selection, to be the maximal degree in \eq{eij_derivative} as $i$ and $j$ run over all pairs of sites:

\begin{definition}\label{def:degree}
The \emph{degree} of the process under weak selection is $D\coloneqq\max_{1\leqslant i,j\leqslant N}D_{ij}$. 
\end{definition}

\begin{moranex}
In this particular process, we have
\begin{align}
e_{ij}\left(\vx\right) &= \frac{x_{i}r+1-x_{i}}{\sum_{k=1}^{N}\left(x_{k}r+1-x_{k}\right)} \frac{1}{N} .
\end{align}
If $r=r\left(\delta\right)$ is a smooth function of $\delta$ with $r\left(0\right) =1$, then
\begin{align}
\frac{d}{d\delta}\Bigg\vert_{\delta =0}e_{ij}\left(\vx\right) &= \frac{1}{N^{2}}r'\left(0\right)\left( x_{i} -\frac{1}{N}\sum_{k=1}^{N}x_{k}\right) .
\end{align}
The degree of the Moran process is thus $D=1$, with $c_{\varnothing}^{ij}=0$ and
\begin{align}
c_{\left\{k\right\}}^{ij} &= 
\begin{cases}
\displaystyle \frac{1}{N^{2}}\left(1-\frac{1}{N}\right) r'\left(0\right)	& \displaystyle k = i , \\
\displaystyle 	& \\
\displaystyle -\frac{1}{N^{3}} r'\left(0\right)	& \displaystyle k\neq i .
\end{cases}
\end{align}
\end{moranex}
Beyond this linear example, in a degree-two process, $\frac{d}{d\delta}\big|_{\delta=0} e_{ij}\left(\vx\right)$ can be represented as a quadratic function of $x_1, \ldots, x_N$, with terms only of the form $c_\varnothing^{ij}$,  $c_{\left\{k\right\}}^{ij}x_k$, or $c_{\left\{h, k\right\}}^{ij}x_hx_k$. In this case, replacement probabilities under weak selection depend on only pairwise statistics of assortment and not on higher-order associations.

We turn now to fixation probabilities. For $\bxi\in\mathbb{B}^{N}$, let $\rho_{A}\left(\bxi\right)$ (resp. $\rho_{B}\left(\bxi\right)$) be the probability that the state $\vA$ (resp. $\vB$) is eventually reached after starting in state $\bxi$. Since states $\vA$ and $\vB$ are absorbing and all other states are transient, we have $\rho_{B}\left(\bxi\right) =1-\rho_{A}\left(\bxi\right)$ for each $\bxi \in \bB^N$.

In the case of neutral drift ($\delta=0$), we let $\pi_{i}$ be the probability of fixation for type $A$ when starting from state $\mathbf{1}_{\{i\}}$; that is, $\pi_i = \rho_{A}\left(\mathbf{1}_{\{i\}}\right)$ when $\delta=0$. Equivalently, $\pi_{i}$ is the probability, under neutral drift, that $i$ is eventually the ancestor of the entire population. These site-specific fixation probabilities are the unique solution to the system of equations (\citealp{allen:PLOSCB:2015}, Theorem 2; \citealp{allen:JMB:2019}, Theorem 7)
\begin{subequations}
\label{eq:pisystem}
\begin{align}
\sum_{j=1}^{N} e_{ij}^{\circ} \pi_{j} &= \sum_{j=1}^{N} e_{ji}^{\circ} \pi_{i} \quad \left(1\leqslant i\leqslant N\right) ; \\
\sum_{i=1}^{N} \pi_{i} &= 1 .
\end{align}
\end{subequations}
The quantity $\pi_{i}$ can be interpreted as the \emph{reproductive value (RV)} of site $i$ \citep{fisher:OUP:1930,taylor:AN:1990,maciejewski:JTB:2014a,allen:JMB:2019}, in that it quantifies the expected contribution of site $i$ to the future gene pool, under neutral drift. For any state $\vx\in\mathbb{B}^{N}$, the RV-weighted frequency, $\widehat{x}\coloneqq\sum_{i=1}^{N}\pi_{i}x_{i}$, is equal to the probability that $A$ becomes fixed under neutral drift when the process starts in state $\vx$ (\citealp{allen:JMB:2019}, Theorem 7).

\begin{moranex}
The reproductive value of every location is $1/N$ due to the fact that the population is unstructured. In a state with $i$ mutants, the fixation probability of $A$ is thus $i/N$.
\end{moranex}
In later examples, we will see that this distribution need not be uniform when the population is spatially structured.

Reproductive values provide a natural weighting for quantities characterizing selection; for example, we define the RV-weighted birth rates and death probabilities to be $\widehat{b}_{i}\left(\vx\right)\coloneqq\sum_{j=1}^{N}e_{ij}\left(\vx\right)\pi_{j}$ and $\widehat{d}_{i}\left(\vx\right)\coloneqq\sum_{j=1}^{N}e_{ji}\left(\vx\right)\pi_{i}$, respectively. In state $\vx$, the change in reproductive-value-weighted frequency of $A$ due to selection, in one step of the process, is
\begin{align}
\label{eq:delhatsel}
\delhatsel\left(\vx\right) &\coloneqq \sum_{i=1}^{N}x_{i}\left(\widehat{b}_{i}\left(\vx\right) -\widehat{d}_{i}\left(\vx\right)\right) = \sum_{i=1}^{N} \pi_{i} \sum_{j=1}^{N} \left( x_{j}-x_{i}\right) e_{ji}\left(\vx\right)
\end{align}
\citep{allen:JMB:2019}. Since $\widehat{b}_{i}^{\circ}=\widehat{d}_{i}^{\circ}$ for $i=1,\dots ,N$, it follows that $\delhatsel^{\circ}\left(\vx\right) =0$ for every $\vx\in\mathbb{B}^{N}$.

\section{Stationary distributions, sojourn times, and coalescence}\label{sec:sojourn}
We are ultimately interested in quantifying (to first order in $\delta$) the probability $\rho_A\left(\bxi\right)$ that type $A$ reaches fixation from initial state $\bxi$. To do so, we will need to quantify the frequency with which the Markov chain visits a given state $\vx \in \bB^N$ prior to absorption in state $\vA$ or $\vB$. We will describe this frequency in two ways: using sojourn times and using the stationary distribution of an amended Markov chain. These two notions are closely connected, as we prove in Proposition \ref{prop:sojourn} below.

We define the sojourn time $t_{\bxi}\left(\vx\right)$, for $\vx \in\mathbb{B}^{N}$, to be the expected number of visits to $\vx$ prior to hitting $\left\{\vA ,\vB\right\}$ when the process begins in state $\bxi\in\mathbb{B}_{\T}^{N}$. These sojourn times $t_{\bxi}\left(\vx\right)$ are uniquely determined by the recurrence relation
\begin{align}
t_{\bxi}\left(\vx\right) &= 
\begin{cases}
\displaystyle 0 & \displaystyle \vx\in\left\{\vA ,\vB\right\} , \\
 & \\
\displaystyle 1 + \sum_{\vy\in\mathbb{B}^{N}}  t_{\bxi}\left(\vy\right) P_{\vy \rightarrow\vx} & \displaystyle \vx = \bxi , \\
 & \\
\displaystyle \sum_{\vy\in\mathbb{B}^{N}}  t_{\bxi}\left(\vy\right) P_{\vy \rightarrow\vx}  & \displaystyle \vx\not\in\left\{\vA ,\vB ,\bxi\right\} .
\end{cases}
\label{eq:trecur}
\end{align}

Since the transition probabilities are continuously differentiable in $\delta$ in a neighborhood of $\delta =0$, so is $t_{\bxi}\left(\vx\right)$.

It is also helpful to consider an amended Markov chain that can ``reset'' in state $\bxi\in\mathbb{B}_{\T}^{N}$ after one of the monoallelic states, $\vA$ or $\vB$, is reached. We introduce a new parameter $u>0$, and define, for $\vx, \vy \in \bB^N$, the amended transition probabilities
\begin{align}
P_{\vx\rightarrow\vy}^{\MSS\left(\bxi\right)} &\coloneqq 
\begin{cases}
\displaystyle u & \displaystyle \vx\in\left\{\vA ,\vB\right\} ,\ \vy =\bxi , \\
& \\
\displaystyle \left(1-u\right) P_{\vx\rightarrow\vy} & \displaystyle \vx\in\left\{\vA ,\vB\right\} ,\ \vy\neq\bxi , \\
& \\
\displaystyle P_{\vx\rightarrow\vy} & \displaystyle \vx\not\in\left\{\vA, \vB\right\} .
\end{cases} \label{eq:MSS_chain}
\end{align}
Above, $P_{\vx\rightarrow\vy}$ refers to the transition probability in the original Markov chain. Thus the amended chain has the same transition probabilities except that, from either of the monoallelic states $\vA$ or $\vB$, there is probability $u$ to transition to state $\bxi$ (see \fig{mutantFixation}). Since $\bxi$ can be any polymorphic state, and since $\bxi$ is the only polymorphic state that can follow $\vA$ or $\vB$, any interpretation of $u$ as a ``mutation'' is likely tenuous from a biological standpoint. However, this amended chain plays an integral technical role in deriving our main results, which in turn do apply under much more realistic assumptions of mutant appearance (discussed in \sect{mutappear}). In particular, the amended chain is clearly aperiodic, and it follows from the {\fix} that it has a single closed communicating class, and all states not in this class are transient. The amended chain therefore has a unique stationary distribution, which we denote by $\left\{\pi_{\MSS\left(\bxi\right)}\left(\vx\right)\right\}_{\vx \in \bB^N}$, the notation $\MSS\left(\bxi\right)$ indicating regeneration into state $\bxi$.

\begin{figure}
	\centering
	\includegraphics[width=0.75\textwidth]{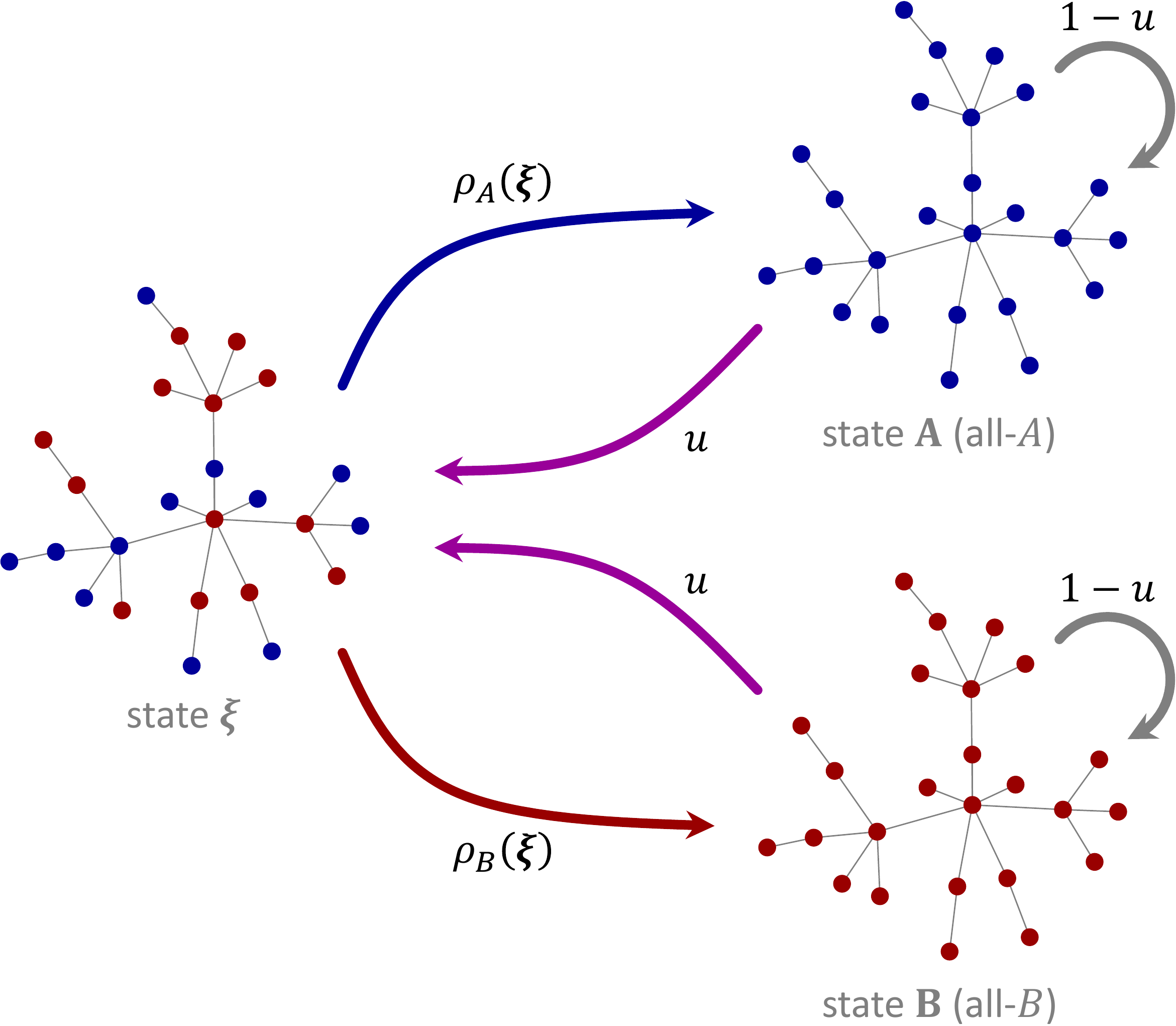}
	\caption{Transitions into a fixed transient state, $\bxi$, following absorption. When starting from a non-monomorphic state, the process will eventually reach one of the two absorbing states (all-$A$ or all-$B$) by the {\fix}. From each absorbing state, the process transitions to $\bxi$ with probability $u\geqslant 0$. This ``artificial'' mutation allows one to focus on the fixation probabilities when the process is started in a fixed initial configuration, $\bxi$, of $A$ and $B$.\label{fig:mutantFixation}}
\end{figure}

Consider now the Markov chain on the monoallelic states whose transition matrix is
\begin{align}
\Lambda \coloneqq \bordermatrix{%
	& \vA & \vB \cr
	\vA &\ \rho_{A}\left(\bxi\right) & \ \rho_{B}\left(\bxi\right) \cr
	\vB &\ \rho_{A}\left(\bxi\right) & \ \rho_{B}\left(\bxi\right) \cr
}.
\end{align}
This chain describes the process in which transitions are first from $\vA$ or $\vB$ to $\bxi$, deterministically, and then from $\bxi$ to $\vA$ with probability $\rho_{A}\left(\bxi\right)$ and to $\vB$ with probability $\rho_{B}\left(\bxi\right)$. This chain is ``embedded'' in the amended chain in the sense that when $u$ is small, the amended chain spends most of its time in $\left\{\vA ,\vB\right\}$, but occasionally it transitions to $\bxi$ before returning to $\left\{\vA ,\vB\right\}$ according to the fixation probabilities. More formally, on a state-by-state basis, the stationary distribution of the embedded chain coincides with $\pi_{\MSS\left(\bxi\right)}$ on the monoallelic states in the limit $u\rightarrow 0$ (\citealp{fudenberg:JET:2006}, Theorem 2), i.e.
\begin{align}
\lim_{u\rightarrow 0} \pi_{\MSS\left(\bxi\right)}\left(\vx\right) &= 
\begin{cases}
\displaystyle \rho_{A}\left(\bxi\right) & \displaystyle \vx =\vA , \\
& \\
& \\ \displaystyle \rho_{B}\left(\bxi\right) & \displaystyle \vx =\vB , \\
& \\
& \\ \displaystyle 0 & \displaystyle \vx\not\in\left\{\vA ,\vB\right\} .
\end{cases} \label{eq:MSS_limit}
\end{align}

The following result, which is key to our methodology, shows that sojourn times of the original chain coincide with the $u$-derivative, at $u=0$, of the stationary distribution for the amended chain:
\begin{proposition}\label{prop:sojourn}
For every non-monoallelic state $\vx\in\mathbb{B}_{\T}^{N}$,
\begin{align}
\frac{d}{du}\Bigg\vert_{u=0}\pi_{\MSS\left(\bxi\right)}\left(\vx\right) &= t_{\bxi}\left(\vx\right) . \label{eq:MSS_sojourn}
\end{align}
\end{proposition}
\begin{proof}
For $\vx\in\mathbb{B}_{\T}^{N}$, stationarity implies
\begin{align}
\pi_{\MSS\left(\bxi\right)}\left(\vx\right) & = \sum_{\vy \in \mathbb{B}^{N}} \pi_{\MSS\left(\bxi\right)}\left(\vy\right) P_{\vy\rightarrow\vx}^{\MSS\left(\bxi\right)} \nonumber \\
& = \sum_{\vy \in \mathbb{B}_{\T}^{N}} \pi_{\MSS\left(\bxi\right)}\left(\vy\right) P_{\vy\rightarrow\vx} 
+ u \delta_{\bxi,\vx} \left( \pi_{\MSS\left(\bxi\right)}\left(\vA\right) + \pi_{\MSS\left(\bxi\right)}\left(\vB\right)  \right),
\end{align}
where the Kronecker symbol $\delta_{\vx,\bxi}$ is equal to 1 if $\vx=\bxi$ and 0 otherwise. Taking the $u$-derivative at $u=0$, and noting that $\lim_{u\rightarrow 0} \left( \pi_{\MSS\left(\bxi\right)}\left(\vA\right) + \pi_{\MSS\left(\bxi\right)}\left(\vB\right)  \right) = 1$ by \eq{MSS_limit}, we obtain
\begin{align}
\frac{d}{du}\Bigg\vert_{u=0}\pi_{\MSS\left(\bxi\right)}\left(\vx\right)
= \sum_{\vy \in \mathbb{B}_{\T}^{N}} \frac{d}{du}\Bigg\vert_{u=0} \pi_{\MSS\left(\bxi\right)}\left(\vy\right) P_{\vy\rightarrow\vx} 
+  \delta_{\bxi,\vx}.
\end{align}
We observe that $\frac{d}{du}\Big\vert_{u=0}\pi_{\MSS\left(\bxi\right)}\left(\vx\right)$ satisfies the same recurrence relation, \eq{trecur}, as $t_{\bxi}\left(\vx\right)$. Since this recurrence relation uniquely defines the times $t_{\bxi}\left(\vx\right)$, we have \eq{MSS_sojourn}.
\end{proof}

It follows immediately from \prop{sojourn} that $t_{\bxi}\left(\mathbb{B}_{\T}^{N}\right)$, the expected time to absorption when starting from state $\bxi$, is equal to $\lim_{u\rightarrow 0}\pi_{\MSS\left(\bxi\right)}\left(\mathbb{B}_{\T}^{N}\right) /u$. With $\E_{\MSS\left(\bxi\right)}\left[\cdot\right]$ denoting expectation with respect to $\pi_{\MSS\left(\bxi\right)}$, we have the following result:
\begin{corollary} 
\label{cor:MSStime}
For any function $\varphi :\mathbb{B}^{N}\rightarrow\mathbb{R}$ with $\varphi\left(\mathbf{A}\right) =\varphi\left(\mathbf{B}\right) =0$,
\begin{align}
\frac{d}{du}\Bigg\vert_{u=0}\E_{\MSS\left(\bxi\right)}\left[\varphi\right] = \sum_{t=0}^{\infty} \E\left[ \varphi\left(\vx^{t}\right) \mid \vx^{0}=\bxi\right] , \label{eq:MSS_series}
\end{align}
and the sum on the right-hand side converges absolutely.
\end{corollary}

\begin{proof}
 $\sum_{t=0}^{\infty}\left|\E\left[ \varphi\left(\vx^{t}\right) \mid \vx^{0}=\bxi\right]\right|$ is bounded by $t_{\bxi}\left(\mathbb{B}_{\T}^{N}\right)\max_{\vx\in\mathbb{B}_{\T}^{N}}\left|\varphi\left(\vx\right)\right|$, so the right-hand side of \eq{MSS_series} converges absolutely. We may therefore rearrange this summation to obtain
\begin{align}
\sum_{t=0}^{\infty} \E\left[ \varphi\left(\vx^{t}\right) \mid \vx^{0}=\bxi\right] &= \sum_{t=0}^\infty \sum_{\vx \in \mathbb{B}_{\T}^{N}} \Prob \left[\vx^{t} = \vx \mid \vx^{0}=\bxi \right] \varphi\left(\vx\right) \nonumber \\
&= \sum_{\vx \in \mathbb{B}_{\T}^{N}} \varphi\left(\vx\right) \sum_{t=0}^{\infty} \Prob \left[\vx^t = \vx \mid \vx^0=\bxi \right] \nonumber \\
&= \sum_{\vx \in \mathbb{B}_{\T}^{N}} \varphi\left(\vx\right) t_{\bxi}\left(\vx\right) \nonumber \\
&= \frac{d}{du}\Bigg\vert_{u=0}\E_{\MSS\left(\bxi\right)}\left[\varphi\right] ,
\end{align}
as desired.
\end{proof}

In light of \cor{MSStime}, we define the operator $\langle \cdot \rangle_{\bxi}$, acting on state functions $\varphi :\mathbb{B}^{N}\rightarrow\mathbb{R}$ with $\varphi\left(\mathbf{A}\right) =\varphi\left(\mathbf{B}\right) =0$, by
\begin{align}
\langle \varphi \rangle_{\bxi} &\coloneqq \frac{d}{du}\Bigg\vert_{u=0}\E_{\MSS\left(\bxi\right)}\left[\varphi\right] = \sum_{t=0}^\infty \E\left[\varphi(\vx^t) \mid \vx^{0}=\bxi\right] .
\end{align}
We will use the notation $\langle \cdot \rangle_{\bxi}^{\circ}$ to indicate that the above expectations are taken under neutral drift ($\delta=0$).

For any function $\varphi: \mathbb{B}^{N}\rightarrow\mathbb{R}$ and any $\widetilde{\alpha}:\left\{1,\dots, N\right\}\rightarrow\left\{1,\dots, N\right\}$, we define $\varphi_{\widetilde{\alpha}}:\mathbb{B}^{N}\rightarrow\mathbb{R}$ by $\varphi_{\widetilde{\alpha}} \left(\vx\right) =\varphi\left(\vx_{\widetilde{\alpha}} \right)$. If $\varphi\left(\vA\right) =\varphi\left(\vB\right) =0$, then
\begin{align}
\left \langle \varphi \right \rangle_{\bxi}^{\circ} &= \varphi\left(\bxi\right) + \sum_{t=0}^\infty \E^{\circ} \left[ \varphi(\vx^{t+1}) \mid \vx^{0}=\bxi \right] \nonumber \\
&= \varphi\left(\bxi\right) + \sum_{t=0}^\infty \sum_{\left(R,\alpha\right)} p_{\left(R,\alpha\right)}^{\circ} \E^{\circ} \left[ \varphi\left(\vx^{t}_{\widetilde{\alpha}}\right) \mid \vx^{0}=\bxi \right] \nonumber \\
&= \varphi\left(\bxi\right) + \sum_{\left(R,\alpha\right)} p_{\left(R,\alpha\right)}^{\circ} \left \langle \varphi_{\widetilde{\alpha}} \right \rangle_{\bxi}^{\circ} ,
\end{align}
which gives the following lemma:
\begin{lemma}\label{lem:RMCexpectation}
For any state function $\varphi :\mathbb{B}^{N}\rightarrow\mathbb{R}$ with $\varphi\left(\vA\right) =\varphi\left(\vB\right) =0$,
\begin{align}
\left \langle \varphi \right \rangle_{\bxi}^{\circ} &= \varphi\left(\bxi\right) + \sum_{\left(R,\alpha\right)} p_{\left(R,\alpha\right)}^{\circ} \left \langle  \varphi_{\widetilde{\alpha}} \right \rangle_{\bxi}^{\circ} . \label{eq:varphiRecurrence}
\end{align}
\end{lemma}

\citet{allen:JMB:2019} introduced the \emph{rare-mutation conditional (RMC)} distribution, defined for a state $\vx$ as $\lim_{u\rightarrow 0} \Prob_{\MSS\left(\bxi\right)}\left[\mathbf{X}=\vx \mid \vx \in \mathbb{B}_{\T}^{N}\right]$, where $\Prob_{\MSS\left(\bxi\right)}\left[\cdot\right]$ denotes probability with respect to $\pi_{\MSS\left(\bxi\right)}$. Here, we show that this distribution can be equated to the fraction of time spent in state $\vx$ out of all transient states:
\begin{corollary} \label{cor:MSS_RMC} For each $\vx \in \mathbb{B}_{\T}^{N}$
\begin{align}
\lim_{u\rightarrow 0} \Prob_{\MSS\left(\bxi\right)} \left[\mathbf{X}=\vx \mid \vx \in \mathbb{B}_{\T}^{N}\right] = \frac{t_{\bxi}\left(\vx\right)}{t_{\bxi}\left(\mathbb{B}_{\T}^{N}\right)}.
\end{align}
\end{corollary}

\begin{proof}
For $\vx \in \mathbb{B}_{\T}^{N}$,
\begin{align}
\lim_{u\rightarrow 0} \Prob_{\MSS\left(\bxi\right)} \left[ \mathbf{X}=\vx \mid \vx \in \mathbb{B}_{\T}^{N} \right] &= \lim_{u\rightarrow 0} \frac{\pi_{\MSS\left(\bxi\right)} \left(\vx\right)}{\pi_{\MSS\left(\bxi\right)} \left(\mathbb{B}_{\T}^{N} \right)} \nonumber \\
&= \frac{ \lim_{u\rightarrow 0} \pi_{\MSS\left(\bxi\right)} \left(\vx\right)/u}{ \lim_{u\rightarrow 0} \pi_{\MSS\left(\bxi\right)} \left(\mathbb{B}_{\T}^{N} \right)/u} \nonumber \\
&= \frac{t_{\bxi}\left(\vx\right)}{t_{\bxi}\left(\mathbb{B}_{\T}^{N}\right)},
\end{align}
by Proposition \ref{prop:sojourn}.
\end{proof}

Finally, we introduce a set-valued Markov chain that will be used in the proof of our main result. The states of this Markov chain are subsets of $\left\{1,\dots ,N\right\}$. From a given state $I \subseteq \left\{1,\dots ,N\right\}$, a new state $I'$ is determined by choosing a replacement event $\left(R,\alpha\right)$ according to the neutral probabilities $p_{\left(R,\alpha\right)}^{\circ}$, and setting $I'=\widetilde{\alpha}\left(I\right)$. This Markov chain, which we denote $\mathcal{C}$, can be understood as a coalescent process \citep{kingman1982coalescent,liggett:S:1985,cox1989coalescing,wakeley2016coalescent}. At each time-step, $\mathcal{C}$ transitions from a set of individuals  $I \subseteq \{1, \ldots, N\}$ to the set $\widetilde{\alpha}\left(I\right)$ of parents of these individuals. (In the case that an individual $i \in I$ is not replaced, the ``parent'' is $i$ itself; that is, $\widetilde{\alpha}\left(i\right) =i$.) Thus, with $\mathcal{C}_{0}=\left\{1,\dots ,N\right\}$, the state of the process after $t$ steps can be understood as the set of ancestors of the current population, at time $t$ before the present.

By the {\fix}, $\mathcal{C}$ has a single closed communicating class consisting only of singleton subsets. (In biological terms, the population's ancestry eventually converges on a common ancestor. The event that $\mathcal{C}$ first reaches a singleton set is called \emph{coalescence}, and the vertex in this singleton set represents the location of the population's most recent common ancestor.)  It follows that $\mathcal{C}$ has a unique stationary distribution concentrated on the singleton subsets. In this stationary distribution, the probability of the singleton set $\{i\}$ in this stationary distribution is given by the reproductive value $\pi_i$.

\section{Fixation probabilities}\label{sec:fixation}
We now prove our main results regarding fixation probabilities. First, we show that the weak-selection expansion of a trait's fixation probability has a particular form:
\begin{theorem} \label{thm:rhoexpand}
For any fixed initial configuration $\bxi\in\mathbb{B}_{\T}^{N}$,
\begin{align}
\label{eq:rhoexpand}
\rho_{A}\left(\bxi\right) & = \widehat{\xi} + \delta\left\langle \frac{d}{d\delta}\Bigg\vert_{\delta =0} \delhatsel \right\rangle_{\bxi}^{\circ} + O\left(\delta^{2}\right) .
\end{align}
\end{theorem}
Theorem~\ref{thm:rhoexpand} generalizes earlier results of \cite{rousset:JTB:2003}, \cite{lessard:JMB:2007}, \cite{chen:AAP:2013}, and \cite{van2015social}. 

\begin{proof}
In the chain defined by \eq{MSS_chain}, the expected change in the reproductive-value-weighted frequency of $A$ in state $\vx$, in one step of the process, is given by
\begin{align}
\delhat_{\MSS\left(\bxi\right)}\left(\vx\right) &= 
\begin{cases}
\displaystyle -u\left(1-\widehat{\xi}\right) & \displaystyle \vx =\vA , \\
& \\
\displaystyle u\widehat{\xi} & \displaystyle \vx =\vB , \\
& \\
\displaystyle \delhatsel\left(\vx\right) & \displaystyle \vx\not\in\left\{\vA ,\vB\right\} .
\end{cases}
\end{align}
Averaging this expected change out over the distribution $\pi_{\MSS\left(\bxi\right)}$ gives
\begin{align}
0 &= \E_{\MSS\left(\bxi\right)}\left[ \delhat_{\MSS\left(\bxi\right)} \right] \nonumber \\
&= \E_{\MSS\left(\bxi\right)}\left[ \delhatsel \right] - u\pi_{\MSS\left(\bxi\right)}\left(\vA\right)\left(1-\widehat{\xi}\right) + u\pi_{\MSS\left(\bxi\right)}\left(\vB\right)\widehat{\xi} .
\end{align}
Differentiating both sides of this equation with respect to $u$ at $u=0$, applying \eq{MSS_limit}, and rearranging, we obtain
\begin{align}
\rho_{A}\left(\bxi\right) = \widehat{\xi} + \left \langle \delhatsel \right \rangle_{\bxi}. \label{eq:diffU}
\end{align}
Since $\delhatsel^{\circ}\left(\vx\right) =0$ for every $\vx\in\mathbb{B}^{N}$, and since the replacement rule is a smooth function of $\delta$ around $0$, we have 
\begin{align}
\left \langle \delhatsel \right \rangle_{\bxi} &= 
\delta \frac{d}{d\delta}\Bigg\vert_{\delta =0} \left \langle \delhatsel \right \rangle_{\bxi} + O\left(\delta^{2}\right) \nonumber \\
&= \delta \left \langle \frac{d}{d\delta}\Bigg\vert_{\delta =0} \delhatsel \right \rangle_{\bxi}^{\circ}+ O\left(\delta^{2}\right) .
\end{align}
The interchange of the $\langle \cdot \rangle_{\bxi}^{\circ}$ operator with the $\delta$-derivative is justified by \cor{MSStime}. Combining this equation with \eq{diffU} completes the proof. 
\end{proof}

Alternatively, Theorem \ref{thm:rhoexpand} can be established by writing
\begin{align}
\rho_{A}\left(\bxi\right) & = \lim_{t\rightarrow \infty} \E \left[\widehat{x}^{t} \mid \vx^{0}=\bxi \right] \nonumber \\
& = \widehat{\xi} + \sum_{t=0}^\infty \E \left[\widehat{x}^{t+1}-\widehat{x}^{t} \mid \vx^{0}=\bxi \right] \nonumber \\
& = \widehat{\xi} + \sum_{t=0}^\infty \E \left[\delhatsel\left(\vx^{t}\right) \mid \vx^{0}=\bxi \right] \nonumber \\
& = \widehat{\xi} + \left \langle \delhatsel \right \rangle_{\bxi} .
\end{align}
This calculation recovers \eq{diffU}, and the rest of the proof follows as above.

Our second main result provides a systematic way to compute the first-order term of the weak-selection expansion, \eq{rhoexpand}:
\begin{theorem}\label{thm:main}
For any fixed initial configuration $\bxi\in\mathbb{B}_{\T}^{N}$,
\begin{align}
\rho_{A}\left(\bxi\right) & = \widehat{\xi} + \delta\left(
\sum_{i=1}^{N} \pi_{i} \sum_{j=1}^{N} \sum_{\substack{I\subseteq\left\{1,\dots ,N\right\} \\ 0\leqslant\left| I\right|\leqslant D_{ji}}} c_{I}^{ji}\left(\eta_{\left\{i\right\}\cup I}^{\bxi}-\eta_{\left\{j\right\}\cup I}^{\bxi}\right) \right) + O\left(\delta^{2}\right) , \label{eq:derivative}
\end{align}
where $\eta^{\bxi}$ is the unique solution to the equations
\begin{subequations}\label{eq:generalRecurrence}
\begin{align}
\label{eq:generalRecurrencea}
\eta_{I}^{\bxi} &= \widehat{\xi}-\bxi_{I} + \sum_{\left(R,\alpha\right)} p_{\left(R,\alpha\right)}^{\circ} \eta_{\widetilde{\alpha}\left(I\right)}^{\bxi} \quad \left(1\leqslant\left| I\right|\leqslant D +1\right) ; \\
\label{eq:generalRecurrenceb}
\sum_{i=1}^{N}\pi_{i} \eta_{\left\{i\right\}}^{\bxi} &= 0 .
\end{align}
\end{subequations}
\end{theorem}
The term $D$ appearing in this system is the degree of the process under weak selection; see \defn{degree}. To simplify notation, in what follows we occasionally drop the bracket notation when using $\eta$ (e.g. $\eta_{i}^{\bxi}$ instead of $\eta_{\left\{i\right\}}^{\bxi}$).
\begin{proof}
Let us define 
\begin{align}
\label{eq:etadef}
\eta_{I}^{\bxi}\coloneqq \left \langle \widehat{x}-\vx_{I} \right \rangle_{\bxi}^{\circ}.
\end{align}
From Theorem \ref{thm:rhoexpand} and \eqsa{eij_derivative}{delhatsel}, we have
\begin{align}
\frac{d}{d\delta}\Bigg\vert_{\delta =0}\rho_{A}\left(\bxi\right) &= \left \langle \frac{d}{d\delta}\Bigg\vert_{\delta =0} \delhatsel \right \rangle_{\bxi}^{\circ} \nonumber \\
& = \sum_{i=1}^{N} \pi_{i} \sum_{j=1}^{N} \sum_{\substack{I\subseteq\left\{1,\dots ,N\right\} \\ 0\leqslant\left| I\right|\leqslant D_{ji}}} c_{I}^{ji} \left \langle \left( x_{j}-x_{i}\right) \vx_{I} \right \rangle_{\bxi}^{\circ} \nonumber \\
& = \sum_{i=1}^{N} \pi_{i} \sum_{j=1}^{N} \sum_{\substack{I\subseteq\left\{1,\dots ,N\right\} \\ 0\leqslant\left| I\right|\leqslant D_{ji}}} c_{I}^{ji} \left( \eta_{\left\{i\right\}\cup I}^{\bxi}-\eta_{\left\{j\right\}\cup I}^{\bxi}\right) ,
\end{align}
which proves \eq{derivative}. \eq{generalRecurrenceb} follows immediately from the definition of $\eta_{I}$ in \eq{etadef}. To obtain \eq{generalRecurrencea} we apply \lem{RMCexpectation}:
\begin{align}
\eta_{I}^{\bxi} &= \left \langle \widehat{x}-\vx_{I} \right \rangle_{\bxi}^{\circ} \nonumber \\
&= \widehat{\xi} - \bxi_{I} + \sum_{\left(R,\alpha\right)} p_{\left(R,\alpha\right)}^{\circ} \left\langle \widehat{x_{\widetilde{\alpha}}}-\vx_{\widetilde{\alpha}\left(I\right)} \right\rangle_{\bxi}^{\circ} \nonumber \\
&= \widehat{\xi} - \bxi_{I} + \sum_{\left(R,\alpha\right)} p_{\left(R,\alpha\right)}^{\circ} \left\langle \widehat{x}-\vx_{\widetilde{\alpha}\left(I\right)} \right\rangle_{\bxi}^{\circ} \nonumber \\
&= \widehat{\xi}-\bxi_{I} +\sum_{\left(R,\alpha\right)}p_{\left(R,\alpha\right)}^{\circ}\eta_{\widetilde{\alpha}\left(I\right)}^{\bxi} ,
\end{align}
where the penultimate line follows from the fact that $\widehat{x}$ is a martingale under neutral drift \citep[][Theorem $7$]{allen:JMB:2019}.

To prove uniqueness of the solution to \eq{generalRecurrence}, we consider the coalescent Markov chain $\mathcal{C}$ defined in the previous section. Let $\mathbf{C}$ be the transition matrix for $\mathcal{C}$, and let $\mathbf{p}$ be its stationary distribution in vector form, which satisfies $\mathbf{p}\left(I\right) =\pi_{i}$ if $I=\left\{i\right\}$ and $\mathbf{p}\left(I\right) =0$ if $\left| I\right|\neq 1$. By uniqueness of the stationary distribution, $\mathbf{C}$ has a simple unit eigenvalue, with corresponding one-dimensional left and right eigenspaces spanned by $\mathbf{p}^T$ and $\mathbf{1}$, respectively. We observe that \eq{generalRecurrencea} can be written in the form $\left(\mathbf{I}-\mathbf{C}\right)\vy =\mathbf{b}$, where $\vy$ has entries $\eta_{I}^{\bxi}$ and $\mathbf{b}$ has entries $\widehat{\xi}-\bxi_I$. We have already exhibited a solution for $\vy$ in \eq{etadef}; call this solution $\vy_0$. By the above remarks about $\mathbf{C}$, the most general solution is $\vy=\vy_0 + K \mathbf{1}$ for an arbitrary scalar $K$. Now we impose \eq{generalRecurrenceb}, which can be written $\mathbf{p}^T\vy=0$. Since $\mathbf{p}^T\vy_0=0$ and $\mathbf{p}^T\mathbf{1}=1$, we must have $K=0$, which leaves $\vy=\vy_0$ as the unique solution.
\end{proof}

\begin{remark}
For a process of degree $D$, calculating $\frac{d}{d\delta}\Big\vert_{\delta =0}\rho_{A}$ involves solving a linear system of size $O\left(N^{D+1}\right)$ (\eq{generalRecurrence}). The complexity of solving for $\pi_{i}$ and $c_{I}^{ij}$ is negligible in comparison. Since the complexity of solving a linear system of $n$ equations is $O\left(n^{3}\right)$, it follows that calculating $\frac{d}{d\delta}\Big\vert_{\delta =0}\rho_{A}$ is $O\left(N^{3\left(D+1\right)}\right)$.
\end{remark}

Since the quantities $\eta_I$ arise as the solution to a system of equations related to the coalescent Markov chain $\mathcal{C}$, it is natural to ask whether the $\eta_I$  have a coalescent-theoretic interpretation. We show in the following sections that, in some cases when the initial state is chosen from a particular  probability distribution, the $\eta_I$ have a natural interpretation as coalescence times or as branch lengths of a coalescent tree. However, these interpretations do not appear to extend to the case of fixation from an arbitrary initial state $\bxi$. 

\section{Stochastic mutant appearance}\label{sec:mutappear}
Having obtained (in Theorems \ref{thm:rhoexpand} and \ref{thm:main}) a weak-selection expansion for fixation probabilities from a particular state $\bxi$, we now generalize to the case that the initial state is sampled from a probability distribution. Specifically, we introduce the probability measures $\mu_{A}$ and $\mu_{B}$, on $\mathbb{B}_{\T}^{N}$, to describe the state of the process after type $A$ or $B$, respectively, has been introduced into the population. We refer to $\mu_{A}$ and $\mu_{B}$ as \emph{mutant-appearance distributions}, although the initial state could just as well arise by some mechanism other than mutation (migration, experimental manipulation, etc.). These distributions can depend on the intensity of selection, $\delta$, and we assume that they are differentiable at $\delta =0$.

Two mutant-appearance distributions often considered in evolutionary models are \emph{uniform initialization} (a single mutant appears at a uniformly chosen site; \citealp{lieberman:Nature:2005,adlam:PRSA:2015}) and \emph{temperature initialization} (a single mutant appears at a site chosen proportionally to the death rate $d_i$; \citealp{allen:PLOSCB:2015,adlam:PRSA:2015}). To define these formally, we define the \emph{complement} $\overline{\vx}$ of a state $\vx\in\mathbb{B}^{N}$ by $\overline{x}_{i}\coloneqq 1-x_{i}$ for $i=1,\dots ,N$.

\begin{example}[Uniform initialization]\label{ex:uniform}

\begin{align}
\mu_{A}\left(\mathbf{1}_{i}\right) &= \mu_{B}\left(\overline{\mathbf{1}}_{i}\right) = \frac{1}{N} \quad \left(1\leqslant i\leqslant N\right) .
\end{align}
\end{example}

\begin{example}[Temperature initialization]\label{ex:temperature}
\begin{subequations}
\begin{align}
\mu_{A}\left(\mathbf{1}_{i}\right) &= \frac{d_{i}\left(\vB\right)}{\sum_{j=1}^{N}d_{j}\left(\vB\right)} \quad \left(1\leqslant i\leqslant N\right) ; \\
\mu_{B}\left(\overline{\mathbf{1}}_{i}\right) &= \frac{d_{i}\left(\vA\right)}{\sum_{j=1}^{N}d_{j}\left(\vA\right)} \quad \left(1\leqslant i\leqslant N\right) .
\end{align}
\end{subequations}
Unlike uniform initialization, temperature initialization opens up the possibility that the mutant-appearance distributions depend on the intensity of selection, $\delta$.
\end{example}

We call a mutant-appearance distribution \emph{symmetric} if it does not distinguish between the two types:

\begin{definition}
We say that $\mu_{A}$ and $\mu_{B}$ are \emph{symmetric} if $\mu_{A}\left(\overline{\bxi}\right) =\mu_{B}\left(\bxi\right)$ for every $\bxi\in\mathbb{B}_{\T}^{N}$.
\end{definition}

Uniform initialization (\ex{uniform}) gives symmetric $\mu_{A}$ and $\mu_{B}$ by definition. If mutant initialization is temperature-based (\ex{temperature}), then $\mu_{A}$ and $\mu_{B}$ are symmetric when $d_{i}\left(\vA\right) =d_{i}\left(\vB\right)$ for $i=1,\dots ,N$. This condition is obviously satisfied under neutral drift ($\delta =0$) but could be violated when $\delta >0$.

The expected fixation probabilities for $A$ and $B$, initialized according to $\mu_{A}$ and $\mu_{B}$, respectively, are
\begin{subequations}
\begin{align}
\E_{\mu_{A}}\left[\rho_{A}\right] &= \sum_{\bxi\in\mathbb{B}_{\T}^{N}} \mu_{A}\left(\bxi\right) \rho_{A}\left(\bxi\right) ; \\
\E_{\mu_{B}}\left[\rho_{B}\right] &= \sum_{\bxi '\in\mathbb{B}_{\T}^{N}} \mu_{B}\left(\bxi '\right) \rho_{B}\left(\bxi '\right) .
\end{align}
\end{subequations}
Since $\rho_{A}^{\circ}\left(\bxi\right) =\widehat{\xi}$ and $\rho_{B}^{\circ}\left(\bxi\right) =1-\widehat{\xi}$ when $\delta =0$, we have $\E_{\mu_{A}}^{\circ}\left[\rho_{A}^{\circ}\right] =\E_{\mu_{A}}^{\circ}\left[\widehat{\xi}\right]$ and $\E_{\mu_{B}}^{\circ}\left[\rho_{B}^{\circ}\right] =1-\E_{\mu_{B}}^{\circ}\left[\widehat{\xi}\right]$. More generally, we have $\E_{\mu_{B}}\left[\rho_{B}\right] =1-\E_{\mu_{B}}\left[\rho_{A}\right]$.

As an immediate consequence of Theorem~\ref{thm:main}, we obtain the following result for the fixation probability of a given type from a given mutant appearance distribution:
\begin{corollary}
\label{cor:mutappear}
If $\mu_{A}$ is a mutant-appearance distribution for $A$, then
\begin{align}
\frac{d}{d\delta}\Bigg\vert_{\delta =0}\E_{\mu_{A}}\left[\rho_{A}\right] &= \frac{d}{d\delta}\Bigg\vert_{\delta =0}\E_{\mu_{A}}\left[\widehat{\xi}\right] + \sum_{i=1}^{N} \pi_{i} \sum_{j=1}^{N} \sum_{\substack{I\subseteq\left\{1,\dots ,N\right\} \\ 0\leqslant\left| I\right|\leqslant D_{ji}}} c_{I}^{ji}\left(\eta_{\left\{i\right\}\cup I}^{\mu_{A}}-\eta_{\left\{j\right\}\cup I}^{\mu_{A}}\right) ,
\end{align}
where $\eta^{\mu_{A}}$ is the unique solution to the equations
\begin{subequations} \label{eq:etamuA}
\begin{align}
\label{eq:etamuAa}
\eta_{I}^{\mu_{A}} &= \E_{\mu_{A}}^{\circ}\left[\widehat{\xi}-\bxi_{I} \right] +\sum_{\left(R,\alpha\right)} p_{\left(R,\alpha\right)}^{\circ} \eta_{\widetilde{\alpha}\left(I\right)}^{\mu_{A}} \quad \left(1\leqslant\left| I\right|\leqslant D+1\right) ; \\
\label{eq:etamuAb}
\sum_{i=1}^{N}\pi_{i}\eta_{i}^{\mu_{A}} &= 0 .
\end{align}
\end{subequations}
Similarly, since $\frac{d}{d\delta}\Big\vert_{\delta =0}\E_{\mu_{B}}\left[\rho_{B}\right] =-\frac{d}{d\delta}\Big\vert_{\delta =0}\E_{\mu_{B}}\left[\rho_{A}\right]$, we have
\begin{align}
\frac{d}{d\delta}\Bigg\vert_{\delta =0}\E_{\mu_{B}}\left[\rho_{B}\right] &= -\frac{d}{d\delta}\Bigg\vert_{\delta =0}\E_{\mu_{B}}\left[\widehat{\xi}\right] -\sum_{i=1}^{N} \pi_{i} \sum_{j=1}^{N} \sum_{\substack{I\subseteq\left\{1,\dots ,N\right\} \\ 0\leqslant\left| I\right|\leqslant D_{ji}}} c_{I}^{ji}\left(\eta_{\left\{i\right\}\cup I}^{\mu_{B}}-\eta_{\left\{j\right\}\cup I}^{\mu_{B}}\right) ,
\end{align}
where $\eta^{\mu_{B}}$ is the unique solution to the equations
\begin{subequations}
\begin{align}
\eta_{I}^{\mu_{B}} &= \E_{\mu_{B}}^{\circ}\left[\widehat{\xi}-\bxi_{I} \right] +\sum_{\left(R,\alpha\right)} p_{\left(R,\alpha\right)}^{\circ} \eta_{\widetilde{\alpha}\left(I\right)}^{\mu_{B}} \quad \left(1\leqslant\left| I\right|\leqslant D+1\right) ; \\
\sum_{i=1}^{N}\pi_{i}\eta_{i}^{\mu_{B}} &= 0 .
\end{align}
\end{subequations}
\end{corollary}

In the case of uniform initialization (\ex{uniform}), we have $\E_{\mu_A}\left[\xi_{i}\right] =1/N$ for all $i\in\left\{1,\dots, N\right\}$. In particular, $\E_{\mu_{A}}^{\circ}\left[\widehat{\xi}-\xi_{i}\right] =0$ for every $i$. Since $\eta^{\mu_A}_{i}=0$ for $i=1,\dots ,N$ is then a solution to \eq{etamuAa}, and since this solution satisfies \eq{etamuAb}, it must be the unique solution to \eq{etamuA} in the case that $\left| I\right| =1$ (see the proof of \tthm{main}). (This argument is used repeatedly in later examples.) Furthermore, $\E_{\mu_{A}}\left[\bxi_{I}\right] =0$ for all $\left| I\right|\geqslant 2$. \eq{etamuA} then simplifies to 
\begin{align}
\eta^{\mu_A}_{I} = 
\begin{cases}
\displaystyle \frac{1}{N} +\sum_{\left(R,\alpha\right)} p_{\left(R,\alpha\right)}^{\circ} \eta^{\mu_A}_{\widetilde{\alpha}\left(I\right)} & \displaystyle 2\leqslant\left| I\right|\leqslant D+1, \\
& \\
\displaystyle 0 & \displaystyle \left| I\right| =1.
\end{cases}
\end{align}
In this case, $\eta^{\mu_A}_{I}$ is equal to $1/N$ times the expected number of steps for the coalescent Markov chain $\mathcal{C}$ to reach a singleton set from initial set $I$. In biological terms, $N\eta^{\mu_A}_I$ is the expected time for the lineages of the individuals in set $I$ to coalesce at a most recent common ancestor.

The overall weak-selection expansion of a trait's fixation probability, in the case of uniform initialization, becomes
\begin{align}
\E_{\mu_A} \left[\rho_A \right] = \frac{1}{N}
+ \delta \sum_{i=1}^{N} \pi_{i} \sum_{j=1}^{N} \sum_{\substack{I\subseteq\left\{1,\dots ,N\right\} \\ 0\leqslant\left| I\right|\leqslant D_{ji}}} c_{I}^{ji}\left(\eta_{\left\{i\right\}\cup I}^{\mu_{A}}-\eta_{\left\{j\right\}\cup I}^{\mu_{A}}\right) 
+ O\left( \delta^2 \right).
\end{align}

\section{Comparing fixation probabilities}\label{sec:comparing}

When evaluating which of two competing types is favored by selection, it is common to compare their fixation probabilities. Specifically, the success of $A$ relative to $B$ is often quantified by whether $\E_{\mu_{A}}\left[\rho_{A}\right] >\E_{\mu_{B}}\left[\rho_{B}\right]$ ($A$ is favored) or $\E_{\mu_{A}}\left[\rho_{A}\right] <\E_{\mu_{B}}\left[\rho_{B}\right]$ ($B$ is favored) \citep{nowak:Nature:2004,allen:JMB:2014,tarnita:AN:2014}. This comparison is natural when the mutant-appearance distributions, $\mu_{A}$ and $\mu_{B}$, are symmetric.

\begin{figure}
	\centering
	\includegraphics[width=0.75\textwidth]{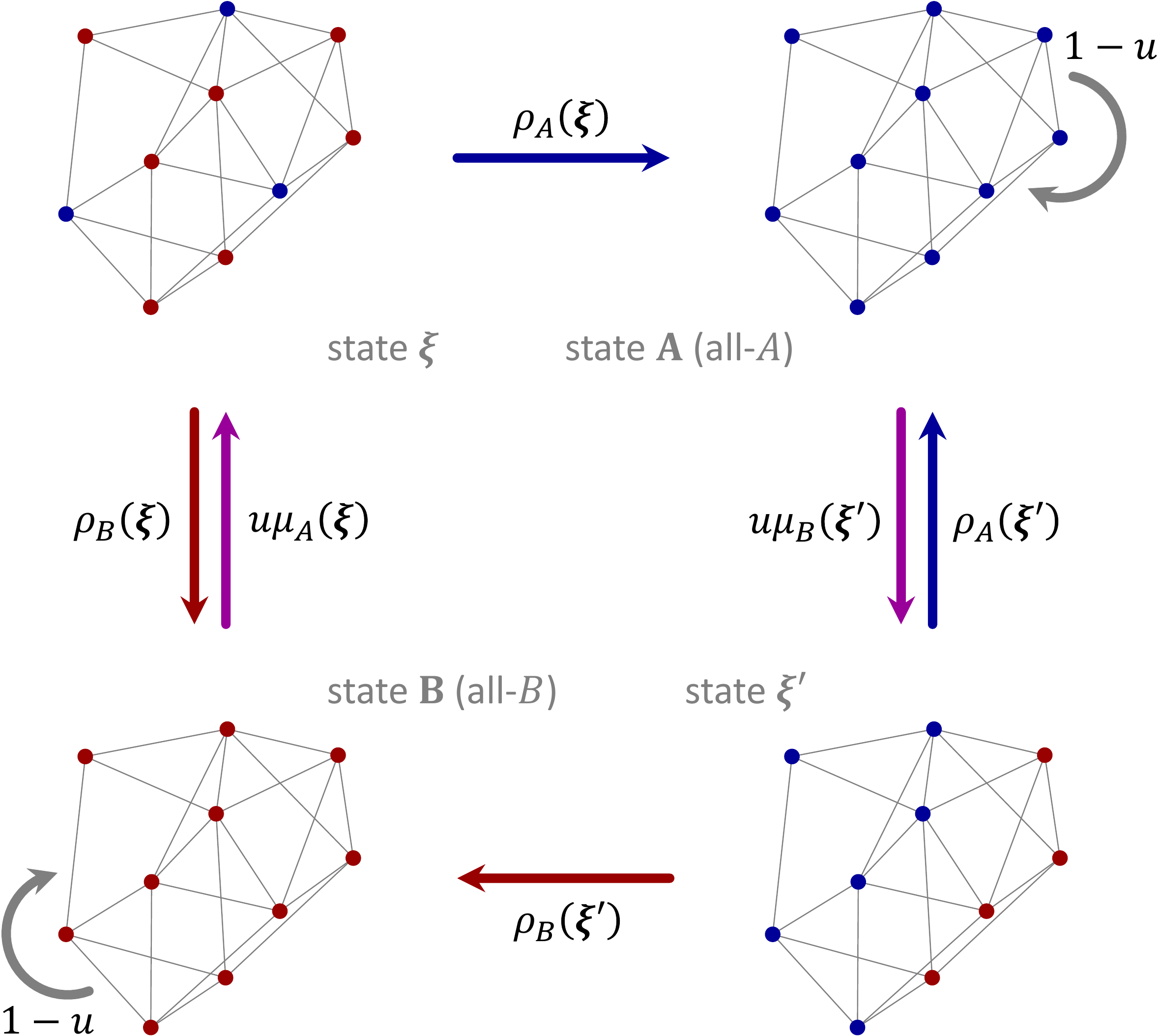}
	\caption{Mutant appearance and fixation in a structured population. In the all-$B$ state, mutants of type $A$ appear based on the mutant-appearance distribution $\mu_{A}$. In the all-$A$ state, mutants of type $B$ appear based on $\mu_{B}$. Once the process transitions into a non-monomorphic state, it will eventually reach one of the two monomorphic states.\label{fig:muAmuB}}
\end{figure}

When $\mu_{A}$ and $\mu_{B}$ are not symmetric, however, it is less natural to compare $\E_{\mu_{A}}\left[\rho_{A}\right]$ to $\E_{\mu_{B}}\left[\rho_{B}\right]$ directly. We therefore derive an alternative measure of success for the asymmetric case. As in Section \ref{sec:sojourn}, we find it helpful to work with an amended Markov chain that possesses a unique stationary distribution. To this end, we suppose that, from the monomorphic state $\vB$, with probability $u$ the next state is chosen from the distribution $\mu_{A}$; otherwise, with probability $1-u$, the state stays $\vB$. Similarly, from the monomorphic state $\vA$, with probability $u$ the next state is chosen from the distribution $\mu_{B}$; otherwise, with probability $1-u$, the state stays $\vB$. The structure of this amended Markov chain is depicted in \fig{muAmuB}. Provided $u>0$, this process has a unique stationary distribution. In analogy to our analysis of mutations into a fixed state in \sect{sojourn}, we denote this stationary distribution by $\pi_{\MSS\left(\mu_{A},\mu_{B}\right)}$.

We now turn to the limit of low mutation. Applying standard results on stationary distributions in this limit  \citep[Theorem 2]{fudenberg:JET:2006}, we have
\begin{align}
 \lim_{u\rightarrow 0} \pi_{\MSS\left(\mu_{A},\mu_{B}\right)}\left(\vx\right) =
 \begin{cases}
 \displaystyle \frac{\E_{\mu_{A}}\left[\rho_{A}\right]}{\E_{\mu_{A}}\left[\rho_{A}\right] +\E_{\mu_{B}}\left[\rho_{B}\right]} & \vx=\vA , \\
 & \\
 \displaystyle \frac{\E_{\mu_{B}}\left[\rho_{B}\right]}{\E_{\mu_{A}}\left[\rho_{A}\right] +\E_{\mu_{B}}\left[\rho_{B}\right]} & \vx=\vB , \\
 & \\
 \displaystyle 0 & \text{otherwise}.
 \end{cases}
\end{align}
Intuitively, for small but positive mutation rates, the process is almost always in state $\vA$ or $\vB$, with the fraction of time spent in $\vA$ converging to $\E_{\mu_{A}}\left[\rho_{A}\right] /\left(\E_{\mu_{A}}\left[\rho_{A}\right] +\E_{\mu_{B}}\left[\rho_{B}\right]\right)$. 

We say that weak selection favors $A$ over $B$ if the fraction of time spent in state $A$ is greater under weak selection than it is in the neutral case:

\begin{definition}
Weak selection favors $A$ over $B$ if
\begin{align}
\label{eq:rhocompare}
\frac{\E_{\mu_{A}}\left[\rho_{A}\right]}{\E_{\mu_{A}}\left[\rho_{A}\right] +\E_{\mu_{B}}\left[\rho_{B}\right]} &> \frac{\E_{\mu_{A}}^{\circ}\left[\rho_{A}^{\circ}\right]}{\E_{\mu_{A}}^{\circ}\left[\rho_{A}^{\circ}\right] +\E_{\mu_{B}}^{\circ}\left[\rho_{B}^{\circ}\right]}
\end{align}
(or, equivalently, $\E_{\mu_{A}}\left[\rho_{A}\right] /\E_{\mu_{B}}\left[\rho_{B}\right] >\E_{\mu_{A}}^{\circ}\left[\rho_{A}^{\circ}\right] /\E_{\mu_{B}}^{\circ}\left[\rho_{B}^{\circ}\right]$) for all sufficiently small $\delta >0$.
\end{definition}

When $\mu_{A}^{\circ}$ and $\mu_{B}^{\circ}$ are symmetric (e.g. temperature or uniform initialization), we have $\E_{\mu_{A}}^{\circ}\left[\rho_{A}^{\circ}\right] =\E_{\mu_{B}}^{\circ}\left[\rho_{B}^{\circ}\right]$, and weak selection favors $A$ relative to $B$ if and only if $\E_{\mu_{A}}\left[\rho_{A}\right] >\E_{\mu_{B}}\left[\rho_{B}\right]$ for all sufficiently small $\delta >0$. For general mutant-appearance distributions, $\mu_{A}$ and $\mu_{B}$, weak selection favors $A$ if 
\begin{align}
    \frac{d}{d\delta}\Bigg\vert_{\delta =0} \frac{\E_{\mu_{A}}\left[\rho_{A}\right]} {\E_{\mu_{A}}\left[\rho_{A}\right] +\E_{\mu_{B}}\left[\rho_{B}\right]}>0.
\end{align}
The left-hand side above has the sign of
\begin{align}
\E_{\mu_{B}}^{\circ}\left[\rho_{B}^{\circ}\right] \frac{d}{d\delta}\Bigg\vert_{\delta =0} \E_{\mu_{A}}\left[\rho_{A}\right]  - 
\E_{\mu_{A}}^{\circ}\left[\rho_{A}^{\circ}\right] \frac{d}{d\delta}\Bigg\vert_{\delta =0} \E_{\mu_{B}}\left[\rho_{B}\right] .
\end{align}
Applying \cor{mutappear}, we obtain the condition
\begin{align}
\E_{\mu_{B}}^{\circ}\left[\rho_{B}^{\circ}\right]\frac{d}{d\delta}\Bigg\vert_{\delta =0}\E_{\mu_{A}}\left[\widehat{\xi}\right] + \E_{\mu_{A}}^{\circ}\left[\rho_{A}^{\circ}\right]\frac{d}{d\delta}\Bigg\vert_{\delta =0}\E_{\mu_{B}}\left[\widehat{\xi}\right] \nonumber \\
 + \sum_{i=1}^{N} \pi_{i} \sum_{j=1}^{N} \sum_{\substack{I\subseteq\left\{1,\dots ,N\right\} \\ 0\leqslant\left| I\right|\leqslant D_{ji}}} c_{I}^{ji}\left(\eta_{\left\{i\right\}\cup I}^{\mu_{AB}}-\eta_{\left\{j\right\}\cup I}^{\mu_{AB}}\right) >0, \label{eq:comparingAB}
\end{align}
where $\eta_{I}^{\mu_{AB}}\coloneqq\E_{\mu_{B}}^{\circ}\left[\rho_{B}^{\circ}\right]\eta_{I}^{\mu_{A}}+\E_{\mu_{A}}^{\circ}\left[\rho_{A}^{\circ}\right]\eta_{I}^{\mu_{B}}$ satisfies the system of equations
\begin{subequations}
\begin{align}
\eta_{I}^{\mu_{AB}} &= \E_{\mu_{A}}^{\circ}\left[\rho_{A}^{\circ}\right]\left(1-\E_{\mu_{B}}^{\circ}\left[\bxi_{I}\right]\right) - \E_{\mu_{B}}^{\circ}\left[\rho_{B}^{\circ}\right]\E_{\mu_{A}}^{\circ}\left[\bxi_{I}\right] \nonumber \\
&\quad + \sum_{\left(R,\alpha\right)} p_{\left(R,\alpha\right)}^{\circ} \eta_{\widetilde{\alpha}\left(I\right)}^{\mu_{AB}} \quad \left(1\leqslant\left| I\right|\leqslant D+1\right) ; \label{eq:muAB} \\
\sum_{i=1}^{N}\pi_{i}\eta_{i}^{\mu_{AB}} &= 0 .
\end{align}
\end{subequations}
\eq{muAB} is due to the facts that $\E_{\mu_{A}}^{\circ}\left[\rho_{A}^{\circ}\right] =\E_{\mu_{A}}^{\circ}\left[\widehat{\xi}\right]$ and $\E_{\mu_{B}}^{\circ}\left[\rho_{B}^{\circ}\right] =1-\E_{\mu_{B}}^{\circ}\left[\widehat{\xi}\right]$.

\begin{example}\label{ex:singlemut}
Consider the case of fixation from a single mutant. Suppose that the initial mutant is located at site $i$ with probability $\mu_{i}$ in both monomorphic states $\vA$ and $\vB$. That is,
\begin{align}
\mu_{A}\left(\mathbf{1}_{i}\right) = \mu_{B}\left(\overline{\mathbf{1}}_{i}\right) = \mu_{i} \quad \left(1\leqslant i\leqslant N\right) .
\end{align}
All states not of the form $\mathbf{1}_{i}$ for $i\in \{1,\ldots,N\}$ have probability zero in $\mu_{A}$, and states not of the form $\overline{\mathbf{1}}_{i}$ have probability zero in $\mu_{B}$. Then we have $\E_{\mu_{A}}^{\circ}\left[\xi_{i}\right] = \mu_{i}$ and $\E_{\mu_{B}}^{\circ}\left[\xi_{i}\right] = 1-\mu_{i}$. For non-singleton $I$, we have $\E_{\mu_{A}}^{\circ}\left[\bxi_{I}\right] = 0$ and $\E_{\mu_{B}}^{\circ}\left[\bxi_{I}\right] = 1-\sum_{i \in I} \mu_{i}$. We also have $ \E_{\mu_{A}}^{\circ}\left[\rho_{A}^{\circ}\right] =  \E_{\mu_{B}}^{\circ}\left[\rho_{B}^{\circ}\right]$ by symmetry, and both $\E_{\mu_{A}}\left[\widehat{\xi}\right]$ and $\E_{\mu_{B}}\left[\widehat{\xi}\right]$ are independent of $\delta$. In this case, weak selection favors $A$ if
\begin{align}
\label{eq:mufixprobcompare}
\sum_{i=1}^{N} \pi_{i} \sum_{j=1}^{N} \sum_{\substack{I\subseteq\left\{1,\dots ,N\right\} \\ 1\leqslant\left| I\right|\leqslant D_{ji}}} c_{I}^{ji}\left(\eta_{\left\{i\right\}\cup I}^{\mu_{AB}}-\eta_{\left\{j\right\}\cup I}^{\mu_{AB}}\right) > 0 ,
\end{align}
where the $\eta_{I}^{\mu_{AB}}$ (upon rescaling by $1/\E_{\mu_{A}}^{\circ}\left[\rho_{A}^{\circ}\right]$) satisfy the simplified recurrence relation
\begin{align}
\label{eq:etamu}
\eta_{I}^{\mu_{AB}} = 
\begin{cases}
\displaystyle \sum_{i \in I} \mu_{i} + \sum_{\left(R,\alpha\right)} p_{\left(R,\alpha\right)}^{\circ} \eta_{\widetilde{\alpha}\left(I\right)}^{\mu_{AB}} & \displaystyle 2\leqslant\left| I\right|\leqslant D+1 , \\
& \\
\displaystyle 0 & \displaystyle \left| I\right| = 1 .
\end{cases}
\end{align}
\end{example}

\section{Application to constant fecundity}
\label{sec:constant}
Constant (or frequency-independent) fecundity refers to the case in which the fecundity (reproductive rate) of an individual depends on only its type and not on the composition of the rest of the population. A common setup has a mutant ($A$) of reproductive rate $r>0$ competing against a resident ($B$) of reproductive rate 1 \citep{lieberman:Nature:2005,broom:PRSA:2008,broom:JSTP:2011,voorhees:PRSA:2013,monk:PRSA:2014,hindersin:PLOSCB:2015,kaveh:RSOS:2015,cuesta2017suppressors,pavlogiannis:CB:2018,moller2019exploring,tkadlec2019population}. The fixation probability of the mutant type is then a function of $r$, which we can write as $\rho_{A}\left(r;\bxi\right)$, where $\bxi\in\mathbb{B}_{\T}^{N}$ is the initial configuration of the mutant and resident.

A common feature of these models is that the fecundity is interpreted in a relative sense, meaning that $r$ quantifies the reproductive rate of $A$ relative to $B$. Consequently, the fixation probabilities of both types are invariant under rescaling the reproductive rates of all individuals. In particular, the fixation probability of a mutant of reproductive rate $r$ competing against a resident of reproductive rate 1 is equal (upon dividing by $r$) to the fixation probability of a mutant of reproductive rate 1 competing against a resident of reproductive rate $r^{-1}$. By interchanging the roles of $A$ and $B$, we see that fixation probabilities satisfy the duality
\begin{align}
\label{eq:rrescale}
\rho_{A}\left(r;\bxi\right) &= 1 - \rho_{A}\left(r^{-1};\overline{\bxi}\right) .
\end{align}

For $r$ close to 1, meaning that the mutation has only a small effect on fecundity, a trait's fixation probability can be analyzed using weak-selection methods such as those considered here \citep{allen2020transient}. To apply these methods, we define the mutant's \emph{selection coefficient} as $s=r-1$. We then obtain an expansion of the fixation probability, $\rho_{A}\left(1+s ;\bxi\right)$, around $s=0$. The selection coefficient $s$ plays a similar role to the selection intensity $\delta$ in the foregoing sections, except that $s$ can also be negative, indicating a disadvantageous mutant.

Taking the $s$-derivative of both sides of \eq{rrescale} at $s=0$, it follows that
\begin{align}
\frac{d}{ds} \Bigg\vert_{s=0} \rho_{A}\left(1+s ;\bxi\right) &= \frac{d}{ds} \Bigg\vert_{s=0} \rho_{A}\left(1+s ;\overline{\bxi}\right) \nonumber \\
&= \frac{1}{2}\frac{d}{ds} \Bigg\vert_{s=0} \left( \rho_{A}\left(1+s ;\bxi\right) + \rho_{A}\left(1+s ;\overline{\bxi}\right) \right) .
\end{align}
Applying \tthm{main} (with $s$ in place of $\delta$), we have the following weak-selection expansion for $\rho_{A}$:
\begin{align}
\label{eq:rho_const}
\rho_{A}\left(1+s ;\bxi\right) &= \widehat{\xi} + \frac{s}{2} \sum_{i=1}^{N} \pi_{i} \sum_{j=1}^{N} \sum_{\substack{I\subseteq\left\{1,\dots ,N\right\} \\ 0\leqslant\left| I\right|\leqslant D_{ji}}} c_{I}^{ji}\left(\widetilde{\eta}_{\left\{i\right\}\cup I}^{\bxi}-\widetilde{\eta}_{\left\{j\right\}\cup I}^{\bxi}\right) + O\left(s^{2}\right) ,
\end{align}
where $\widetilde{\eta}_{I}^{\bxi} \coloneqq  \eta_{I}^{\bxi} + \eta_{I}^{\overline{\bxi}}$. These $\widetilde{\eta}_{I}^{\bxi}$ are the unique solution to the recurrence relation
\begin{align}
\widetilde{\eta}_{I}^{\bxi} = 
\begin{cases}
\displaystyle 1-\left(\bxi_{I} + \overline{\bxi}_{I}\right) + \sum_{\left(R,\alpha\right)} p_{\left(R,\alpha\right)}^{\circ} \widetilde{\eta}_{\widetilde{\alpha}\left(I\right)}^{\bxi} & \displaystyle 2\leqslant\left| I\right|\leqslant D+1 , \\
& \\
\displaystyle 0 & \displaystyle \left| I\right| =1 .
\end{cases}
\end{align}
Note that $\bxi_{I}+\overline{\bxi}_{I}$ is equal to 1 if all individuals in $I$ have the same type in $\bxi$ (i.e. $\xi_{i}=\xi_{j}$ for all $i,j \in I$) and 0 otherwise. In particular, $\xi_{i}+\overline{\xi}_{i}=1$ for all $i=1,\dots ,N$, which is why $\widetilde{\eta}_{I}^{\bxi}=0$ for $|I|=1$.

If the initial state $\bxi$ is drawn from a mutant-appearance distribution, $\mu_{A}$, we then have
\begin{align}
\frac{d}{ds}\Bigg\vert_{s =0}\E_{\mu_{A}}\left[\rho_{A}\right] &= \frac{d}{ds}\Bigg\vert_{s =0}\E_{\mu_{A}}\left[\widehat{\xi}\right] 
+ \frac{1}{2} \sum_{i=1}^{N} \pi_{i} \sum_{j=1}^{N} \sum_{\substack{I\subseteq\left\{1,\dots ,N\right\} \\ 0\leqslant\left| I\right|\leqslant D_{ji}}} c_{I}^{ji}\left(\widetilde{\eta}_{\left\{i\right\}\cup I}^{\mu_{A}}-\widetilde{\eta}_{\left\{j\right\}\cup I}^{\mu_{A}}\right) ,
\end{align}
where
\begin{align}
\widetilde{\eta}^{\mu_{A}}_I & = 
\begin{cases}
\displaystyle 1-\E_{\mu_A}^{\circ} \left[\bxi_{I}+\overline{\bm{\xi}}_{I}\right] +\sum_{\left(R,\alpha\right)} p_{\left(R,\alpha\right)}^{\circ} \widetilde{\eta}_{\widetilde{\alpha}\left(I\right)}^{\mu_A} & \displaystyle 2\leqslant\left| I\right|\leqslant D+1 , \\
& \\
\displaystyle 0 & \displaystyle \left| I\right| =1.
\end{cases}
\end{align}
In particular, for the mutant-appearance distribution in \ex{singlemut}, in which the initial state has a single mutant whose location $i$ is chosen with probability $\mu_i$ (independent of $r$), we have
\begin{align}
\E_{\mu_{A}} \left[ \rho_{A} \right] &= \sum_{i=1}^N \pi_{i}\mu_{i} + 
\frac{s}{2}  \sum_{i=1}^{N} \pi_{i} \sum_{j=1}^{N} \sum_{\substack{I\subseteq\left\{1,\dots ,N\right\} \\ 0\leqslant\left| I\right|\leqslant D_{ji}}} c_{I}^{ji}\left(\widetilde{\eta}_{\left\{i\right\}\cup I}^{\mu_{A}}-\widetilde{\eta}_{\left\{j\right\}\cup I}^{\mu_{A}}\right) + O\left(s^{2}\right) , \label{eq:cf_muA}
\end{align}
where, owing to the fact that $\E_{\mu_A}^{\circ} \left[\bxi_{I}+\overline{\bm{\xi}}_{I}\right] =\E_{\mu_A}^{\circ} \left[\overline{\bm{\xi}}_{I}\right] =1-\sum_{i\in I}\mu_{i}$ when $\left| I\right| >1$,
\begin{align}
\label{eq:etamu2}
\widetilde{\eta}^{\mu_{A}}_I & = 
\begin{cases}
\displaystyle \sum_{i \in I} \mu_{i} + \sum_{\left(R,\alpha\right)} p_{\left(R,\alpha\right)}^{\circ} \widetilde{\eta}_{\widetilde{\alpha}\left(I\right)}^{\mu_A} & \displaystyle 2\leqslant\left| I\right|\leqslant D+1 , \\
& \\
\displaystyle 0 & \displaystyle \left| I\right| =1 .
\end{cases}
\end{align}
\eq{etamu2} is identical to the recurrence relation for $\eta^{\mu_{A}}_I$ in \eq{etamu}.

\subsection{Constant fecundity on graphs}
\label{sec:const_fecund_graphs}
Let us now suppose that the population structure is represented by an undirected, unweighted graph. Each vertex is occupied by one individual. As above, we consider a mutant type, $A$, of fecundity $r=1+s$, competing against a resident type, $B$, of fecundity 1. A robust research program aims to elucidate the effects of graph structure on the mutant's fixation probability \citep{lieberman:Nature:2005,broom:PRSA:2008,broom:JSTP:2011,voorhees:PRSA:2013,monk:PRSA:2014,hindersin:PLOSCB:2015,kaveh:RSOS:2015,cuesta2017suppressors,pavlogiannis:CB:2018,moller2019exploring,tkadlec2019population}.

The edge weight between vertices $i$ and $j$ is denoted $w_{ij}$ ($=w_{ji}$). We define the \emph{weighted degree} of vertex $i$ as $w_{i}\coloneqq\sum_{j=1}^{N}w_{ij}$. One can define a natural random walk on this graph, moving from vertex $i$ to vertex $j$ with probability $p_{ij}^{\left(1\right)}\coloneqq w_{ij}/w_{i}$ (where the superscript indicates that this probability is for one step in the random walk). More generally, the probability of moving from $i$ to $j$ in $n$ steps of this random walk is $p_{ij}^{\left(n\right)}\coloneqq\left(p^{\left(1\right)}\right)_{ij}^{n}$, i.e. entry $\left(i,j\right)$ of the $n$th power of the transition matrix $p^{\left(1\right)}$.

\subsubsection{death-Birth updating}
\label{sec:const_fecund_dB}
Under the \emph{death-Birth} process \citep{ohtsuki:Nature:2006,hindersin:PLOSCB:2015,allen2020transient,tkadlec2020limits}, an individual is first chosen, uniformly-at-random from the population, to be replaced (``death''). A neighbor is then chosen, with probability proportional to the product of edge weight and fecundity, to produce an offspring that fills the vacancy (``Birth''). The term ``Birth'' is capitalized here to emphasize the fact that selection influences this step \citep{hindersin:PLOSCB:2015}.

For this process, the probability that the offspring of $i$ replaces the occupant of $j$ is
\begin{align}
e_{ij}\left(1+s ;\vx\right) &= \frac{1}{N} \frac{w_{ij}\left(\left(1+ s\right) x_{i}+1-x_{i}\right)}{\sum_{k=1}^{N}w_{kj}\left(\left(1+ s\right) x_{k} +1-x_{k}\right)} = \frac{p_{ji}^{\left(1\right)}}{N} \frac{1+s x_{i}}{1+s \sum_{k=1}^{N}p_{jk}^{\left(1\right)} x_{k}} .
\end{align}
Differentiating this expression with respect to $s$ at $s=0$ gives
\begin{align}
\frac{d}{ds}\Bigg\vert_{s =0}e_{ij}\left(1+s ;\vx\right) &= \sum_{k=1}^{N} c_{k}^{ij} x_{k} ,
\end{align}
where
\begin{align}
c_{k}^{ij} &= 
\begin{cases}
\displaystyle \frac{p_{ji}^{\left(1\right)}}{N} \left( 1 - p_{ji}^{\left(1\right)} \right) & \displaystyle k=i , \\
& \\
\displaystyle -\frac{p_{ji}^{\left(1\right)}}{N} p_{jk}^{\left(1\right)} & \textrm{otherwise} .
\end{cases}
\end{align}
This process therefore has degree $D=1$.

Under neutral drift, $e_{ij}^{\circ}=p_{ji}^{\left(1\right)}/N$, and solving \eq{pisystem} yields $\pi_{i}=w_{i}/\sum_{j=1}^{N}w_{j}$ as the reproductive value of vertex $i$ under death-Birth updating \citep[see also][]{maciejewski:JTB:2014a,allen:Nature:2017,allen:JMB:2019}. Using \eq{rho_const} and the reversibility property $\pi_{i}p_{ij}^{\left(1\right)}=\pi_{j}p_{ji}^{\left(1\right)}$, a series of simplifications gives the following result:

\begin{proposition}
For the death-Birth process on a weighted graph, with constant mutant fecundity $r=1+s$, the fixation probability from arbitrary starting configuration $\bxi$ can be expanded under weak selection as
\begin{align}
\rho_{A}\left(1+s ;\bxi\right) = \widehat{\xi} + \frac{s}{2N} \sum_{i,j=1}^{N} \pi_{i} p_{ij}^{\left(2\right)} \widetilde{\eta}_{ij}^{\bxi} + O\left(s^{2}\right) ,
\end{align}
where the terms $\widetilde{\eta}_{ij}^{\bxi}$ arise as the unique solution to the recurrence relation
\begin{align}
\widetilde{\eta}_{ij}^{\bxi} = 
\begin{cases}
\displaystyle \frac{N}{2}\left( {\xi}_{i}+{\xi}_{j} - 2{\xi}_{i}{\xi}_{j} \right) + \frac{1}{2}
\sum_{k=1}^{N} \left( p_{ik}^{\left(1\right)}\widetilde{\eta}_{kj}^{\bxi} + p_{jk}^{\left(1\right)}\widetilde{\eta}_{ik}^{\bxi} \right)
& \displaystyle i \neq j , \\
& \\
\displaystyle 0 & \displaystyle i=j .
\end{cases}
\end{align}
\end{proposition}

The factor of $2$ above is related to the fact that we focus on only those replacement events that influence $i$ and $j$.

For a single mutant initialized randomly as in \ex{singlemut}, 
 \eq{cf_muA} becomes
\begin{align}
\label{eq:constant_dB_mu}
\E_{\mu_{A}} \left[ \rho_{A} \right] &= \sum_{i=1}^{N} \pi_{i} \mu_{i} +  \frac{s}{2N} \sum_{i,j=1}^{N} \pi_{i} p_{ij}^{\left(2\right)} \widetilde{\eta}_{ij}^{\mu_{A}} + O\left(s^{2}\right) ,
\end{align}
where
\begin{align}
\label{eq:constant_dB_mu_eta}
\widetilde{\eta}_{ij}^{\mu_{A}} =
\begin{cases}
\displaystyle N\left(\frac{\mu_{i}+\mu_{j}}{2}\right) + \frac{1}{2}\sum_{k=1}^{N}\left(p_{ik}^{\left(1\right)}\widetilde{\eta}_{kj}^{\mu_{A}} + p_{jk}^{\left(1\right)}\widetilde{\eta}_{ik}^{\mu_{A}} \right) & \displaystyle i \neq j , \\
 & \\
\displaystyle 0 &\displaystyle i=j .
\end{cases}
\end{align}
\eqsc{constant_dB_mu}{constant_dB_mu_eta} generalize a result of \cite{allen2020transient}, which pertained to the case of uniform initialization, i.e. $\mu_i=1/N$ for all $i$.

\subsubsection{Birth-death updating}
In the \emph{Birth-death} process \citep{lieberman:Nature:2005,hindersin:PLOSCB:2015}, also known as the \emph{invasion process} \citep{antal:PRL:2006}, an individual $i$ is selected to reproduce with probability proportional to fecundity; the offspring of $i$ replaces $j$ with probability $p_{ij}^{\left(1\right)}$.

Letting $\left|\vx\right|\coloneqq\sum_{i=1}^{N}x_{i}$ denote the abundance of type $A$ in state $\vx$, the probability that $i$ replaces $j$ in this state is
\begin{align}
e_{ij}\left(1+s ;\vx\right) &= \frac{\left(1+ s\right) x_{i} + 1-x_{i}}{\sum_{k=1}^{N}\left(\left(1+ s\right) x_{k}+1-x_{k}\right)} p_{ij}^{\left(1\right)} = \frac{1+ sx_{i}}{N+ s\left|\vx\right|} p_{ij}^{\left(1\right)} .
\end{align}
Differentiating this expression with respect to $s$ at $s=0$ gives
\begin{align}
\frac{d}{ds}\Bigg\vert_{s =0}e_{ij}\left(1+s ;\vx\right) &=  \sum_{k=1}^{N} c_{k}^{ij} x_{k} ,
\end{align}
where
\begin{align}
c_{k}^{ij} &= 
\begin{cases}
\displaystyle \frac{N-1}{N^{2}}p_{ij}^{\left(1\right)} & \displaystyle k=i , \\
& \\
\displaystyle -\frac{1}{N^{2}}p_{ij}^{\left(1\right)} & \textrm{otherwise} .
\end{cases}
\end{align}
Under neutral drift, $e_{ij}^{\circ}=p_{ij}^{\left(1\right)}/N$, and \eq{pisystem} yields a reproductive value of  $\pi_{i}=w_{i}^{-1}/\sum_{j=1}^{N}w_{j}^{-1}$ for Birth-death updating \citep[see also][]{maciejewski:JTB:2014a,allen:Nature:2017,allen:JMB:2019}. A series of simplifications based on \eq{rho_const} and the relation $\pi_{i}p_{ji}^{\left(1\right)}=\pi_{j}p_{ij}^{\left(1\right)}$ gives the following result:

\begin{proposition}
For the Birth-death process on a weighted graph, with constant mutant fecundity $r=1+s$, the fixation probability from arbitrary starting configuration $\bxi$ can be expanded under weak selection as
\begin{align}
 \rho_{A}(1+s) = \widehat{\xi}  +  \frac{s}{2N} \sum_{i,j=1}^{N} \pi_{i} p_{ji}^{\left(1\right)} \widetilde{\eta}_{ij}^{\bxi} + O\left(s^{2}\right),
\end{align}
where the terms $\widetilde{\eta}_{ij}^{\bxi}$ arise as the unique solution to
\begin{align}
\widetilde{\eta}_{ij}^{\bxi} = 
\begin{cases}
\displaystyle \frac{N\left( {\xi}_{i}+{\xi}_{j} - 2{\xi}_{i}{\xi}_{j} \right) + \sum_{k=1}^{N} \left( p_{ki}^{\left(1\right)} \widetilde{\eta}_{kj}^{\bxi} +  p_{kj}^{\left(1\right)} \widetilde{\eta}_{ik}^{\bxi}\right)}{\sum_{k=1}^{N}\left(p_{ki}^{\left(1\right)}+p_{kj}^{\left(1\right)}\right)} & \displaystyle i \neq j , \\
& \\
\displaystyle 0 & \displaystyle i=j .
\end{cases}
\end{align}
\end{proposition}

The factor $\sum_{k=1}^{N}\left(p_{ki}^{\left(1\right)}+p_{kj}^{\left(1\right)}\right)$ is analogous to the factor of $2$ in the corresponding expression for death-Birth updating (due to considering the effects of a replacement rule on only $i$ and $j$). However, death is not necessarily uniform under Birth-death updating, which results in a slightly more complicated scaling factor.

For the mutant-appearance distribution of \ex{singlemut}, \eq{cf_muA} simplifies to
\begin{align}
\E_{\mu_{A}} \left[ \rho_{A} \right] &= \sum_{i=1}^{N} \pi_{i} \mu_{i} +  \frac{s}{2N} \sum_{i,j=1}^{N} \pi_{i} p_{ji}^{\left(1\right)} \widetilde{\eta}_{ij}^{\mu_{A}} + O\left(s^{2}\right) ,
\end{align}
where 
\begin{align}
\widetilde{\eta}_{ij}^{\mu_{A}} = 
\begin{cases}
\displaystyle \frac{N\left( \mu_i+\mu_j \right) + \sum_{k=1}^{N} \left( p_{ki}^{\left(1\right)} \widetilde{\eta}_{kj}^{\mu_{A}} +  p_{kj}^{\left(1\right)} \widetilde{\eta}_{ik}^{\mu_{A}}\right)}{\sum_{k=1}^{N}\left(p_{ki}^{\left(1\right)}+p_{kj}^{\left(1\right)}\right)} & \displaystyle i \neq j , \\
& \\
\displaystyle 0 & \displaystyle i=j .
\end{cases}
\end{align}

\section{Application to evolutionary game theory}\label{sec:examples}
We now move from constant fecundity to evolutionary games (frequency-dependent fecundity) in structured populations \citep{blume:GEB:1993,nowak:Nature:1992,ohtsuki:Nature:2006,szabo:PR:2007,nowak:PTRSB:2009}. In this setting, individuals interact with one another and receive a net payoff based on the types (strategies) of those with whom they interact. In state $\vx\in\mathbb{B}^{N}$, we let $u_{i}\left(\vx\right)$ denote the payoff (or utility) of player $i$. This payoff is converted into relative fecundity, $F_{i}$, by letting $F_{i}\left(\vx\right) =\exp\left\{\delta u_{i}\left(\vx\right)\right\}$, where $\delta>0$ is the selection intensity parameter. (An alternative convention, $F_{i}\left(\vx\right) =1+\delta u_{i}\left(\vx\right)$, is equivalent under weak selection, since both satisfy $F_{i}\left(\vx\right)=1+\delta u_i\left(\vx\right) + O\left(\delta^{2}\right)$.) The replacement rule then depends directly on the fecundity vector, $\vF\in\left(0,\infty\right)^{N}$, i.e. $e_{ij}\left(\vx\right) =e_{ij}\left(\vF\left(\vx\right)\right)$. Furthermore, under weak selection, there is no loss of generality in assuming that the state-to-fecundity mapping, $\vx\mapsto\vF\left(\vx\right)$, is deterministic \citep[see][]{mcavoy:NHB:2020}. Therefore, for every $i,j=1,\dots ,N$, we can write
\begin{align}
\frac{d}{d\delta}\Bigg\vert_{\delta =0}e_{ij}\left(\vx\right) &= \sum_{k=1}^{N} \left(\frac{\partial}{\partial F_{k}}\Bigg\vert_{\vF =\mathbf{1}} e_{ij}\left(\vF\right)\right) u_{k}\left(\vx\right) .
\end{align}
Since the fecundity derivative $\frac{\partial}{\partial F_{k}}\Big\vert_{\vF =\mathbf{1}} e_{ij}\left(\vF\right)$ does not depend on $\vx$, the degree of the process under weak selection is controlled by the utility functions $\left\{u_{i}\left(\vx\right)\right\}_{i=1}^{N}$. Let $m_{k}^{ij}\coloneqq\frac{\partial}{\partial F_{k}}\Big\vert_{\vF =\mathbf{1}} e_{ij}\left(\vF\right)$ be the marginal effect of the fecundity of $k$ on $i$ replacing $j$ \citep{mcavoy:NHB:2020}, and suppose that individual $k$'s payoff is $u_{k}\left(\vx\right) =\sum_{I\subseteq\left\{1,\dots ,N\right\}}p_{I}^{k}\vx_{I}$. We then have
\begin{align}
\frac{d}{d\delta}\Bigg\vert_{\delta =0}e_{ij}\left(\vx\right) &= \sum_{k=1}^{N} m_{k}^{ij} u_{k}\left(\vx\right) = \sum_{I\subseteq\left\{1,\dots ,N\right\}} \left( \sum_{k=1}^{N} m_{k}^{ij} p_{I}^{k} \right) \vx_{I} ,
\end{align}
which, by the uniqueness of the representation of \eq{eij_derivative}, gives $c_{I}^{ij}=\sum_{k=1}^{N}m_{k}^{ij}p_{I}^{k}$. Therefore, generically, the degree of the process coincides with the maximal degree of the payoff functions when each payoff function is viewed as a multi-linear polynomial in $x_{1},\dots ,x_{N}$ \citep[see also][]{ohtsuki:PTRSB:2014,mcavoy:JMB:2016}.

\subsection{Additive games} \label{sec:additive}
Additive games are a special class of games for which the conditions to be favored by selection can be written in a simplified form. An evolutionary game is additive if its payoff function, $u_{i}\left(\vx\right)$, is of degree one (linear) in $\vx$ for every $i=1,\dots ,N$. In this case we can write
\begin{align}
\frac{d}{d\delta}\Bigg\vert_{\delta =0}e_{ij}\left(\vx\right) &= c_{0}^{ij} + \sum_{k=1}^{N} c_{k}^{ij}x_{k}
\end{align}
for some collection of constants $c_{k}^{ij}$ with $i,j,k \in \{1,\ldots,N\}$. If we further assume that the all-$B$ state has the same replacement probabilities as the neutral process---that is, $p_{\left(R,\alpha\right)}\left(\vB\right) =p_{\left(R,\alpha\right)}^{\circ}$ for all $\left(R,\alpha\right)$ and all $\delta >0$---it then follows that $c^{ij}_0=0$ for all $i,j$. In this case, \tthm{main} gives
\begin{align}
\frac{d}{d\delta}\Bigg\vert_{\delta =0}\E_{\mu_{A}}\left[\rho_{A}\right] &= \frac{d}{d\delta}\Bigg\vert_{\delta =0}\E_{\mu_{A}}\left[\widehat{\xi}\right] + \sum_{i,j,k=1}^{N} \pi_{i} c_{k}^{ji}\left(\eta_{ik}^{\mu_{A}}-\eta_{jk}^{\mu_{A}}\right) ,
\end{align}
where the terms $\eta^{\mu_{A}}_{ij}$ arise as the unique solution to the equations
\begin{subequations}
\begin{align}
\eta_{ij}^{\mu_{A}} &= \E_{\mu_{A}}^{\circ}\left[\widehat{\xi}-\xi_{i}\xi_{j}\right] +\sum_{\left(R,\alpha\right)} p_{\left(R,\alpha\right)}^{\circ} \eta_{\widetilde{\alpha}\left(i\right)\widetilde{\alpha}\left(j\right)}^{\mu_{A}} \quad \left(1\leqslant i,j\leqslant N\right) ; \\
\sum_{i=1}^{N}\pi_{i}\eta_{ii}^{\mu_{A}} &= 0 .
\end{align}
\end{subequations}

Let us now consider the mutant-appearance distributions of \ex{singlemut}; i.e.~a single mutant appears at site $i$ with probability $\mu_{i}$. The weak-selection fixation probability of type $A$ can then be calculated as
\begin{align}
    \E_{\mu_{A}}\left[\rho_{A}\right] = \sum_{i=1}^N \pi_i \mu_i + \delta \sum_{i,j,k=1}^{N} \pi_{i} c_{k}^{ji}\left(\eta_{ik}^{\mu_{A}}-\eta_{jk}^{\mu_{A}}\right)
    + O\left(\delta^{2}\right) ,
\end{align}
with $\eta^{\mu_{A}}_{ij}$ as above. Furthermore, the condition for weak selection to favor $A$, \eq{mufixprobcompare}, becomes
\begin{align}
\sum_{i,j,k=1}^{N} \pi_{i} c_{k}^{ji}\left(\eta_{ik}^{\mu_{AB}}-\eta_{jk}^{\mu_{AB}}\right) > 0 , \label{eq:simplifiedLinearCondition}
\end{align}
where the terms $\eta_{ij}^{\mu_{AB}}$ arise as the unique solution to
\begin{align}\label{eq:etaSystemLinear}
\eta_{ij}^{\mu_{AB}} =
\begin{cases}
\displaystyle \mu_{i} + \mu_{j} + \sum_{\left(R,\alpha\right)}p_{\left(R,\alpha\right)}^{\circ}\eta_{\widetilde{\alpha}\left(i\right)\widetilde{\alpha}\left(j\right)}^{\mu_{AB}} & \displaystyle i\neq j , \\
& \\
\displaystyle 0 & \displaystyle i = j  .
\end{cases}
\end{align}

\subsection{Graph-structured populations with death-Birth updating}
Evolutionary games on a graphs provide well-known models for social interactions in structured populations \citep{santos2005scale,ohtsuki:Nature:2006,taylor:Nature:2007,szabo:PR:2007,santos:Nature:2008,chen:AAP:2013,debarre:NC:2014,allen:Nature:2017,allen2019isothermal,mcavoy:NHB:2020}. As in \sect{const_fecund_graphs}, we suppose that the population structure is represented by a weighted, undirected graph with adjacency matrix $\left(w_{ij}\right)_{i,j=1}^{N}$. We adopt the notation of \sect{const_fecund_graphs} for the weighted degree $w_i \coloneqq \sum_{j=1}^{N} w_{ij}$ and the step probability $p_{ij}^{\left(1\right)} \coloneqq w_{ij}/w_i$.

As in \sect{const_fecund_dB}, we consider death-Birth updating: an individual is chosen uniformly at random for death, and then a neighboring individual is chosen, with probability proportional to the product of fecundity and edge weight, to reproduce into the vacancy. With this update rule, the marginal probability that $i$ transmits an offspring to $j$ is
\begin{align}
e_{ij}\left(\vx\right) &= \frac{1}{N} \frac{F_{i}\left(\vx\right) w_{ij}}{\sum_{\ell =1}^{N} F_{\ell}\left(\vx\right) w_{\ell j}} .
\end{align}
Differentiating with respect to $\delta$ at $\delta =0$ gives
\begin{align}
\nonumber
\frac{d}{d\delta}\Bigg\vert_{\delta =0} e_{ij}\left(\vx\right) &= \frac{1}{N} \frac{u_{i}\left(\vx\right) w_{ij} w_{j}-w_{ij}\sum_{\ell =1}^{N}u_{\ell}\left(\vx\right) w_{\ell j}}{w_{j}^{2}}\\
\label{eq:deiju}
& = \frac{p_{ji}^{\left(1\right)}}{N} \left(u_i\left(\vx\right) - \sum_{\ell=1}^N p_{j \ell}^{\left(1\right)} u_\ell\left(\vx\right) \right).
\end{align}

We now focus on a particular game called the \emph{donation game} \citep{sigmund:PUP:2010}, in which type $A$ pays a cost of $c$ to donate $b$ to every neighbor; type $B$ donates nothing and pays no cost. For $b>c>0$, this is a special case of the prisoner's dilemma, with $A$ playing the role of cooperators and $B$ playing the role of defectors.

There are two main conventions for aggregating the payoffs received from these game interactions \citep{maciejewski:PLoSCB:2014}. The first is to take the edge-weighted sum of the payoffs received from all others; this method is called \emph{accumulated payoffs}, and leads to $u_{i}\left(\vx\right) =-w_{i}cx_{i}+\sum_{k=1}^{N} w_{ik}bx_{k}$. The second is to take the edge-weighted average (i.e. to normalize the sum by the weighted degree); this method is called \emph{averaged payoffs} and leads to $u_{i}\left(\vx\right) =-cx_{i}+b\sum_{k=1}^{N} p_{ik}^{\left(1\right)}x_{k}$. In both cases, the game is additive according to the definition of the previous subsection.

Substituting the respective payoff functions into \eq{deiju} yields 
\begin{align}
    \frac{d}{d\delta}\Bigg\vert_{\delta =0} e_{ij}\left(\vx\right) 
= \sum_{k=1}^{N} c_{k}^{ij}x_{k},
\end{align}
where 
\begin{align}
\label{eq:cijk_averaged}
c_{k}^{ij} = 
\begin{cases}
\displaystyle \frac{p_{ji}^{\left(1\right)}}{N}
\left(-c\left(p_{ik}^{\left(0\right)}-p_{jk}^{\left(1\right)} \right) +b\left( p_{ik}^{\left(1\right)}-p_{jk}^{\left(2\right)}\right) \right) & \text{(averaged)} , \\
& \\
\displaystyle \frac{p_{ji}^{\left(1\right)}}{N}
\left(-c\left(p_{ik}^{\left(0\right)}w_i - p_{jk}^{\left(1\right)} w_k \right) 
+b\left( w_{ik}-\sum_{\ell=1}^N p_{j\ell}^{\left(1\right)} w_{\ell k} \right) \right) & \text{(accumulated)} .
\end{cases}
\end{align}
Above, we have used $p_{ij}^{\left(n\right)}$ to denote the probability that an $n$-step random walk from $i$ terminates at $j$; note in particular that $p_{ij}^{\left(0\right)}$ equals 1 if $i=j$ and 0 otherwise.  We further define $\eta_{\left(n\right)}^{\bxi} \coloneqq \sum_{i,j=1}^{N}\pi_{i}p_{ij}^{\left(n\right)}\eta_{ij}^{\bxi}$ as the expectation of $\eta_{ij}^{\bxi}$ when $i$ and $j$ are sampled from the two ends of a stationary $n$-step random walk, where $n \geqslant 0$.  We can then state the following result:

\begin{proposition}
\label{prop:dBgames}
For the donation game on an arbitrary weighted, connected graph, with death-Birth updating, the fixation probability of cooperators from an arbitrary configuration $\bxi$ can be expanded under weak selection as 
\begin{align}
\label{eq:db_heterogeneous_avg}
\rho_{A}\left(\bxi\right) = \widehat{\xi} + \frac{\delta}{N} \left( -c \eta^{\bxi}_{\left(2\right)} + b \left(\eta^{\bxi}_{\left(3\right)}- \eta^{\bxi}_{\left(1\right)} \right)  \right)  + O\left(\delta^{2}\right),
\end{align}
for averaged payoffs, and 
\begin{align}
\rho_{A}\left(\bxi\right) 
= \widehat{\xi} +  \frac{\delta}{N} &\Bigg( 
-c \sum_{i,j=1}^{N} \pi_{i} \left(p_{ij}^{\left(2\right)} - p_{ij}^{\left(0\right)} \right) w_{j} \eta_{ij}^{\bxi} \nonumber \\
&\qquad +b\sum_{i,j,k=1}^{N} \pi_{i} \left(p_{ij}^{\left(2\right)} - p_{ij}^{\left(0\right)} \right) w_{jk} \eta_{ik}^{\bxi}  \Bigg) + O\left(\delta^{2}\right), \label{eq:db_heterogeneousFP}
\end{align}
for accumulated payoffs. In both cases, the terms $\eta_{ij}^{\bxi}$ arise as the unique solution to 
\begin{subequations}\label{eq:dbTau}
\begin{align}
\eta_{ij}^{\bxi} &= \frac{N}{2}\left(\widehat{\xi}-\xi_{i}\xi_{j}\right) + \frac{1}{2}\sum_{k=1}^{N} \left( p_{ik}^{\left(1\right)}\eta_{kj}^{\bxi} + p_{jk}^{\left(1\right)}\eta_{ik}^{\bxi} \right) \quad \left(i \neq j\right) ; \\
\eta_{ii}^{\bxi} &= N\left(\widehat{\xi}-\xi_{i}\right) + \sum_{j=1}^{N}p_{ij}^{\left(1\right)}\eta_{jj}^{\bxi} ; \\
\label{eq:dbTausum}
\sum_{i=1}^{N}\pi_{i}\eta_{ii}^{\bxi} &= 0 .
\end{align}
\end{subequations}
\end{proposition}

\begin{proof} \tthm{main} gives
\begin{align}
\rho_{A}\left(\bxi\right) = \widehat{\xi} + \delta \sum_{i,j,k=1}^{N} \pi_{i} c_{k}^{ji}\left(\eta_{ik}^{\bxi}-\eta_{jk}^{\bxi}\right) + O\left(\delta^{2}\right) ,
\end{align}
where the $\eta_{ij}^{\bxi}$ are the unique solution to \eq{dbTau}.  The result then follows from applying \eq{cijk_averaged} and simplifying using the reversibility property $\pi_i p^{\left(n\right)}_{ij} = \pi_j p^{\left(n\right)}_{ji}$, noting that $\eta^{\bxi}_{\left(0\right)} = \sum_{i=1}^N \pi_i \eta^{\bxi}_{ii} = 0$ by \eq{dbTausum}.
\end{proof}

\prop{dBgames} generalizes one of the main results of \cite{allen:Nature:2017} (who considered only uniform initialization) to the case of an arbitrary initial state.

\subsubsection{Homogeneous (regular) graphs}
In the case of a regular graph, we can obtain the weak-selection expansion of fixation probabilities in closed form. Suppose the graph is unweighted (meaning each edge weight is either 0 or 1), has no self-loops ($w_{ii}=0$ for each $i$) and is regular of degree $d$ ($w_i=d$ for all $i$). For regular graphs, accumulated and averaged payoffs are equivalent upon rescaling all payoffs by a factor of $d$; we consider averaged payoffs here.

Noting that $\pi_i=1/N$ for all $i$, and $\widehat{\xi}=\left|\bxi\right| /N$, for regular graphs, \eq{dbTau} for $\eta^{\bxi}_{ij}$ can be written as
\begin{subequations}
\begin{align}
\eta_{ij}^{\bxi} &= \frac{1}{2}\left(\left|\bxi\right| -N\xi_{i}\xi_{j}\right) + \frac{1}{2}\sum_{k=1}^{N} \left( p_{ik}^{\left(1\right)}\eta_{kj}^{\bxi} + p_{jk}^{\left(1\right)}\eta_{ik}^{\bxi} \right) \\
\nonumber & \qquad + \frac{\delta_{ij}}{2} \left( \left|\bxi\right| -N\xi_{i} + 2\sum_{k=1}^N p_{ik}^{\left(1\right)} \left(\eta_{kk}^{\bxi}-\eta_{ik}^{\bxi}\right)\right),\\
\label{eq:etasum_regular}
\sum_{i=1}^{N}\eta_{ii}^{\bxi} &= 0,
\end{align}
\end{subequations}
where $\delta_{ij}$ is the Kronecker delta function. Multiplying by $\frac{1}{N} p_{ij}^{\left(n\right)}$ and summing over $i$ and $j$ leads to a recurrence relation for $\eta^{\bxi}_{(n)}$:
\begin{align}
\label{eq:etanrecur}
\eta^{\bxi}_{(n+1)} - \eta^{\bxi}_{\left(n\right)} &= 
\frac{1}{2} \sum_{i,j=1}^N p_{ij}^{\left(n\right)} \xi_i \xi_j - \frac{1}{2} \left|\bxi\right| \nonumber \\
&\qquad + \frac{1}{2N}\sum_{i=1}^N p_{ii}^{\left(n\right)}\left(N\xi_i-\left|\bxi\right| + 2\sum_{k=1}^N p_{ik}^{\left(1\right)} \left(\eta_{ik}^{\bxi}-\eta_{kk}^{\bxi}\right)\right).
\end{align}
We now let $n\to \infty$, taking a running average in the case that the random walk is periodic. In this limit, $p_{ij}^{\left(n\right)}$ converges to $\pi_j = 1/N$ for each $i$ and $j$. Simplifying and applying \eq{etasum_regular}, \eq{etanrecur} then becomes
\begin{align}
0 & = \frac{1}{2N} \sum_{i,j=1}^N \xi_i\xi_j - \frac{1}{2} \left|\bxi\right| + \frac{1}{2N^2} \sum_{i=1}^N \left(N \xi_i - \left|\bxi\right|\right) + \frac{1}{N^{2}}\left(\sum_{i,k=1}^N p_{ik}^{\left(1\right)} \eta_{ik}^{\bxi}-\sum_{k=1}^N\eta_{kk}^{\bxi}\right) \nonumber \\
& = \frac{1}{2N} \left|\bxi\right|^2 - \frac{1}{2} \left|\bxi\right| + \frac{1}{N} \eta^{\bxi}_{\left(1\right)} ,
\end{align}
which leads to
\begin{align}
    \eta^{\bxi}_{\left(1\right)} = \frac{1}{2}\left|\bxi\right| \left(N-\left|\bxi\right|\right).
\end{align}
Applying \eq{etanrecur}, and noting that $p_{ii}^{\left(2\right)} = 1/d$ for each $i$, we see that
\begin{subequations}
\begin{align}
\eta^{\bxi}_{\left(2\right)} & =  \frac{1}{2} \left( \left|\bxi\right|\left(N-\left|\bxi\right| -1\right) +  \sum_{i,j=1}^N p_{ij}^{\left(1\right)} \xi_i \xi_j \right),\\
\eta^{\bxi}_{\left(3\right)} & = \frac{1}{2} \left( \left|\bxi\right|\left(\frac{d+1}{d} \left(N-\left|\bxi\right|\right)-2\right) +  \sum_{i,j=1}^N \left(p_{ij}^{\left(1\right)} +p_{ij}^{\left(2\right)} \right)\xi_i \xi_j \right).
\end{align}
\end{subequations}
Substituting into \eq{db_heterogeneous_avg}, we obtain the following closed-form result (a corollary to \prop{dBgames}):
\begin{corollary}
\label{cor:dBgamesregular}
For the donation game on an unweighted regular graph with death-Birth updating, the fixation probability of cooperators from arbitrary initial configuration $\bxi$ can be expanded under weak selection as
\begin{align}
\rho_{A}\left(\bxi\right) = \frac{\left|\bxi\right|}{N} + &\frac{\delta}{2N} \Bigg( -c \left(\left|\bxi\right| \left(N-\left|\bxi\right| -1\right) + \sum_{i,j=1}^N p_{ij}^{\left(1\right)} \xi_i \xi_j\right) \nonumber \\
&\qquad +b \left( \left|\bxi\right| \left(\frac{N-\left|\bxi\right|}{d}-2\right) +  \sum_{i,j=1}^N \left(p_{ij}^{\left(1\right)} +p_{ij}^{\left(2\right)} \right)\xi_i \xi_j \right)  \Bigg)  + O\left(\delta^{2}\right) . \label{eq:main_chen2016}
\end{align}
\end{corollary}

\cor{dBgamesregular} is equivalent to the main result of \citet{chen:SR:2016}. In particular, when the initial state contains only a single type $A$ individual ($\left|\bxi\right| =1$), we have $\sum_{i,j=1}^N p_{ij}^{\left(1\right)} \xi_i \xi_j=0$ and $\sum_{i,j=1}^N p_{ij}^{\left(1\right)} \xi_i \xi_j=\sum_{i=1}^N p_{ii}^{\left(2\right)} \xi_i = 1/d$, leading to 
\begin{align}
\rho_{A}\left(\bxi\right) = \frac{1}{N} + \frac{\delta}{2N} \left( -c \left(N-2\right) +b \left( \frac{N}{d}-2  \right) \right) + O\left(\delta^{2}\right).
\end{align}
This result holds regardless of which vertex contains the initial type $A$ individual, as was first proven by \cite{chen:AAP:2013}.

\subsection{Comparing population structures}\label{sec:comparingPopulationStructures}
A large body of literature is devoted to the question of whether---and to what extent---population structure can promote the evolution of cooperation \citep{nowak:Nature:1992,hauert:Nature:2004,santos2005scale,ohtsuki:Nature:2006,taylor:Nature:2007,santos:Nature:2008,nowak:PTRSB:2009,tarnita:JTB:2009,debarre:NC:2014,allen:Nature:2017}. The donation game with $b>c>0$ provides an elegant model for studying this question \citep{sigmund:PUP:2010}: Type $A$ (representing cooperation) pays cost $c$ to give benefit $b$ to its partners; type $B$ pays no costs and gives no benefits. Population structures can then be compared according to whether or not they increase $A$'s chance of becoming fixed, depending on the benefit, $b$, and cost, $c$. Each population structure has a ``critical benefit-to-cost ratio,'' $\left(b/c\right)^{\ast}$, such that weak selection increases $A$'s fixation probability if and only if $\left(b/c\right)^{\ast}>0$ and $b/c>\left(b/c\right)^{\ast}$ \citep{ohtsuki:Nature:2006,nowak:PTRSB:2009,allen:Nature:2017}. 

A lower critical benefit-to-cost can then be interpreted as ``better for the evolution of cooperation'' \citep{nathanson:PLOSCB:2009}. Such quantities have also been used to formally order population structures. For example, \citet{pena:JRSI:2016} state that ``two different models of spatial structure and associated evolutionary dynamics can be unambiguously compared by ranking their relatedness or structure coefficients: the greater the coefficient, the less stringent the conditions for cooperation to evolve. Hence, different models of population structure can be ordered by their potential to promote the evolution of cooperation in a straightforward way.'' While there is indeed an unambiguous comparison of population structures based on critical benefit-to-cost ratios, a comparison based on which is ``better for the evolution of cooperation'' is more subtle.

Consider the donation game with accumulated payoffs on a graph, with death-Birth updating and uniform mutant-appearance distribution. By \eq{db_heterogeneousFP},
\begin{align}
\E_{\textrm{unif}} \left[ \rho_{A} \right]
= \frac{1}{N} +  \frac{\delta}{N} &\Bigg( 
-c \sum_{i,j=1}^{N} \pi_{i} p_{ij}^{\left(2\right)} w_{j} \eta_{ij}^{\textrm{unif}} \nonumber \\
&\qquad +b\sum_{i,j,k=1}^{N} \pi_{i} \left(p_{ij}^{\left(2\right)} - p_{ij}^{\left(0\right)} \right) w_{jk} \eta_{ik}^{\textrm{unif}} \Bigg) + O\left(\delta^{2}\right) ,
\end{align}
where 
\begin{align}
\eta_{ij}^{\textrm{unif}} =
\begin{cases}
\displaystyle \frac{1}{2}\left(1 + \sum_{k=1}^{N} \left( p_{ik}^{\left(1\right)}\eta_{kj}^{\textrm{unif}}+p_{jk}^{\left(1\right)}\eta_{ik}^{\textrm{unif}} \right) \right)
& \displaystyle i\neq j , \\
& \\
\displaystyle 0 & \displaystyle i =j .
\end{cases}
\end{align}
The critical benefit-to-cost ratio is
\begin{align}
\left(\frac{b}{c}\right)^\ast
= \frac{\sum_{i,j=1}^{N} \pi_{i} p_{ij}^{\left(2\right)} w_{j} \eta_{ij}^{\textrm{unif}}}{\sum_{i,j,k=1}^{N} \pi_{i} \left(p_{ij}^{\left(2\right)} - p_{ij}^{\left(0\right)} \right) w_{jk} \eta_{ik}^{\textrm{unif}}}.
\end{align}
In \fig{comparing}, we apply this result to first compute the critical benefit-to-cost ratios for two heterogeneous population structures of size $N=50$. Specifically, we give examples of graphs $\Gamma_{1}$ (\fig{comparing}\textbf{A}) and $\Gamma_{2}$ (\fig{comparing}\textbf{B}) such that $0<\left(b/c\right)_{\Gamma_{1}}^{\ast}<\left(b/c\right)_{\Gamma_{2}}^{\ast}$, which means that the condition for cooperation to be favored on $\Gamma_{1}$ is less strict than that of $\Gamma_{2}$. However, this ranking alone does not imply that $\Gamma_{1}$ is unambiguously better for the evolution of cooperation than $\Gamma_{2}$. For example, when $b=10$ and $c=1$, which corresponds to $b/c>\left(b/c\right)_{\Gamma_{1}}^{\ast},\left(b/c\right)_{\Gamma_{2}}^{\ast}$, weak selection boosts the fixation probability of cooperators on $\Gamma_{2}$ more than it does on $\Gamma_{1}$, based on the magnitudes of the first-order effects of selection. Therefore, for this particular cooperative social dilemma, $\Gamma_{2}$ more strongly supports the evolution of cooperation than $\Gamma_{1}$. It follows that the critical benefit-to-cost ratio provides only part of the story when comparing two population structures based on their abilities to support the emergence of cooperation.

\begin{figure}
	\centering
	\includegraphics[width=0.9\textwidth]{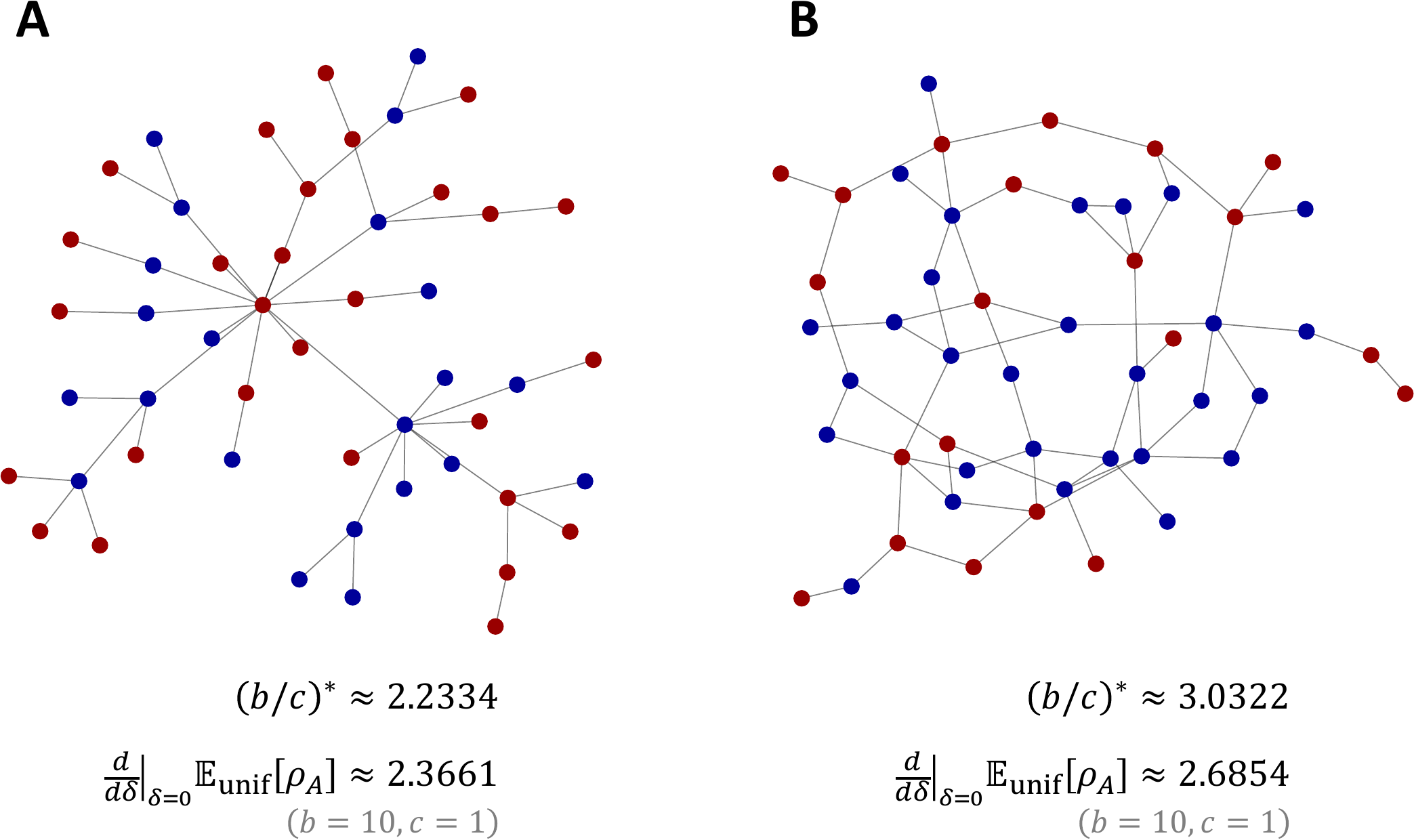}
	\caption{Two heterogeneous population structures of size $N=50$ evolving based on death-Birth updating. The population structure depicted in \textbf{B} is unequivocally better for the evolution of cooperation than that of \textbf{A} when $b=10$ and $c=1$ because selection results in a greater improvement to a rare cooperator's fixation probability. This result holds despite the fact that the critical benefit-to-cost ratio of \textbf{B} is greater than it is for \textbf{A}, which means that the condition under which cooperators can thrive on \textbf{B} is stricter than that of \textbf{A}. It follows that a comparison between population structures based on which one is better for the evolution of cooperation cannot be made based on the critical ratios alone.\label{fig:comparing}}
\end{figure}

\section{Discussion}
In this study, we have analyzed the fixation probability of a mutant type under weak selection, for a broad class of evolutionary models and arbitrary initial conditions. The main result, \tthm{main}, gives a first-order expansion of the fixation probability of $A$, $\rho_{A}\left(\bxi\right)$, in the selection intensity $\delta$, for any initial configuration, $\bxi$. This expansion has three main ingredients: \emph{(i)} reproductive value, $\pi_{i}$, which quantifies the expected contribution of $i$ to future generations; \emph{(ii)} neutral sojourn times, $\eta_{I}^{\bxi}$, which may be interpreted in terms of the mean number of steps in which all individuals in $I$ have type $A$ prior to absorption, given that the initial state of the population is $\bxi$; and \emph{(iii)} Fourier coefficients, $c_{I}^{ij}$, of first-order effects of the probability that $i$ replaces $j$ in one update step.

It follows from \tthm{main} that the complexity of calculating this first-order expansion is $O\left(N^{3\left(D+1\right)}\right)$, where $D$ is the degree of the process. Actually, by the work of \citet{legall:IEEE:2012}, this complexity is (in theory) $O\left(N^{2.373\left(D+1\right)}\right)$. This bound can be further improved in some cases by taking into account structural properties of the population, as we observed in the case of death-Birth updating on a regular graph. In any case, for fixed degree $D$, the system size exhibits polynomial growth in $N$, whereas the number of states in the evolutionary process grows exponentially in $N$.

The neutral sojourn times, $\eta_{I}^{\bxi}$ and variants thereof, play a central role in our method. Their interpretation is therefore a question of interest. From \eq{etadef}, we can see that $-\eta_{I}^{\bxi}=\left \langle \vx_{I} - \widehat{x} \right \rangle_{\bxi}^{\circ}$ is a measure of the tendency for all individuals in $I$ have type $A$, under neutral drift from initial state $\bxi$. In the case of a uniform initial distribution $\mu_A$ (\ex{uniform}), $\eta_{I}^{\mu_A}$ is proportional to the expected time for the coalescent process $\mathcal{C}$ to reach a singleton set (coalesce) starting from set $I$.

The utility of our framework is illustrated by the application to evolutionary dynamics on graphs in Sections \ref{sec:constant} and \ref{sec:examples}. In particular, \prop{dBgames} provides the weak-selection expansion of fixation probabilities for the donation game with arbitrary graph and initial configuration. This result unifies and generalizes the main results of \cite{chen:SR:2016} (who considered only  regular graphs) as well as \cite{allen:Nature:2017} (who considered only uniform initialization). 

From these results, one can derive many of the well-known results on critical benefit-to-cost ratios for cooperation to be favored in social dilemmas \citep{ohtsuki:Nature:2006,taylor:Nature:2007,chen:AAP:2013,allen:EMS:2014,fotouhi:NHB:2018}. Moreover, they provide more information than just \emph{when} weak selection favors a particular trait; they also determine \emph{how much}, based on the magnitude of $\frac{d}{d\delta}\Big\vert_{\delta =0}\rho_{A}$, which can lead to more nuanced comparisons of population structures based on their ability to promote a trait (\sect{comparingPopulationStructures}). The magnitude of $\frac{d}{d\delta}\Big\vert_{\delta =0}\rho_{A}$ has been explored considerably less than its sign, and our results allow this question to be explored for quite a large class of evolutionary update rules, population structures, and initial configurations.

Our results on fixation probabilities apply to finite populations of a given size. It would be interesting to connect these results to the considerable body of theory for large populations \citep{kimura:G:1962,roze2003selection,traulsen2006stochasticity,cox2016evolutionary,chen2018wright}, by analyzing the large-population ($N \to \infty$) asymptotics of our results such as \eq{derivative}. However, a number of challenges arise. First, unless one places a bound on the degree of the process, the number of terms in first-order part of \eq{derivative} grows exponentially with $N$. Second, since \eq{derivative} is itself an expansion for weak selection ($\delta \ll 1$), one must consider the relationship between $N$ and $\delta$ as $N \to \infty$ and $\delta \to 0$.  These two limits are known to be non-interchangeable, even in the relatively simple case of two-player games in a well-mixed population \citep{sample2017limits}. However, neither of these challenges appears insurmountable, and addressing them is an important goal for future work.

\section*{Acknowledgments}
We thank Krishnendu Chatterjee, Joshua Plotkin, Qi Su, and John Wakeley for helpful discussions and comments on earlier drafts. We are also grateful to the anonymous referees for many valuable suggestions. This work was supported by the Army Research Laboratory (grant W911NF-18-2-0265), the John Templeton Foundation (grant 61443), the National Science Foundation (grant DMS-1715315), the Office of Naval Research (grant N00014-16-1-2914), and the Simons Foundation (Math+X Grant to the University of Pennsylvania).

\begin{longtable}{cp{10cm}c}
\caption{Glossary of Notation\label{tab:notationtable}} \\
\hline
Symbol & Description & Introduced\\
\hline
\endhead
\hline \endfoot
$\vA$ & Monomorphic state in which all individuals have type $A$ & \sect{modeling_assumptions} \\
$\alpha$ & Parentage map in a replacement event & \sect{modeling_assumptions} \\
$\widetilde{\alpha}$ & Extension of the parentage map, $\alpha$, to $\left\{1,\dots ,N\right\}$ & \sect{modeling_assumptions} \\
$\vB$ & Monomorphic state in which all individuals have type $B$ & \sect{modeling_assumptions} \\
$\mathbb{B}$ & Boolean domain $\left\{0,1\right\}$ ($1$ for $A$, $0$ for $B$) & \sect{modeling_assumptions} \\
$\mathbb{B}^{N}$ & Set of configurations of types in the population & \sect{modeling_assumptions} \\
$\mathbb{B}_{\T}^{N}$ & Set of non-monomorphic configurations of types in the population, i.e. $\mathbb{B}^{N}-\left\{\vA ,\vB\right\}$ & \sect{modeling_assumptions} \\
$b_{i}\left(\vx\right)$ & Expected offspring number of $i$ in state $\vx$ & \sect{quantifying_selection} \\
$\left(b/c\right)^{\ast}$ & Critical benefit-to-cost ratio for cooperation to evolve & \sect{comparingPopulationStructures} \\
$c_{I}^{ij}$ & Fourier coefficients of $\frac{d}{d\delta}\Big\vert_{\delta =0}e_{ij}\left(\vx\right)$ & \sect{quantifying_selection} \\
$\delta$ & Selection strength & \sect{modeling_assumptions} \\
$d_{i}\left(\vx\right)$ & Death probability of $i$ in state $\vx$ & \sect{quantifying_selection} \\
$\delsel\left(\vx\right)$ & Expected change in the frequency of $A$ due to selection & \sect{quantifying_selection} \\
$D_{ij}$ & Fourier degree of $\frac{d}{d\delta}\Big\vert_{\delta =0}e_{ij}\left(\vx\right)$ as a multi-linear polynomial & \sect{quantifying_selection} \\
$D$ & Fourier degree of the evolutionary process at $\delta =0$ & \sect{quantifying_selection} \\
$e_{ij}\left(\vx\right)$ & Marginal probability that $i$ transmits its offspring to $j$ in state $\vx$ & \sect{quantifying_selection} \\
$\mu_{A}$, $\mu_{B}$ & Mutant-appearance distributions of types $A$ and $B$, respectively & \sect{mutappear} \\
$\mu_{i}$ & Probability that a single mutant appears at location $i$ & \sect{comparing} \\
$m_{k}^{ij}$ & Marginal effect of the fecundity of $k$ on $i$ replacing $j$ & \sect{examples} \\
$N$ & Number of individuals (population size) & \sect{modeling_assumptions} \\
$p_{\left(R,\alpha\right)}\left(\vx\right)$ & Probability of choosing replacement event $\left(R,\alpha\right)$ in state $\vx$ & \sect{modeling_assumptions} \\
$\pi_{i}$ & Reproductive value of location $i$ & \sect{quantifying_selection} \\
$\pi_{\MSS\left(\bxi\right)}$ & Mutation-selection stationary distribution (which depends on a fixed configuration, $\bxi$) & \sect{sojourn} \\
$p_{ij}^{\left(n\right)}$ & Probability of moving from vertex $i$ to vertex $j$ in $n$ steps of a random walk on a graph & \sect{const_fecund_graphs} \\
$R$ & Set of replaced positions in a replacement event & \sect{modeling_assumptions} \\
$r$ & Relative reproductive rate of a mutant in a constant-fecundity process & \sect{constant} \\
$\rho_{A}$, $\rho_{B}$ & Fixation probabilities of types $A$ and $B$, respectively & \sect{quantifying_selection} \\
$s$ & Selection coefficient of a mutant in a constant-fecundity process & \sect{constant} \\
$u$ & Probability of regenerating a fixed transient state, $\bxi$, following $\vA$ or $\vB$ & \sect{sojourn} \\
$u_{i}\left(\vx\right)$ & Payoff to $i$ in state $\vx$ in an evolutionary game & \sect{examples} \\
$w_{ij}$ & Edge weight between vertices $i$ and $j$ in a graph & \sect{const_fecund_graphs} \\
$w_{i}$ & Degree of vertex $i$ in a graph & \sect{const_fecund_graphs} \\
$\vx$ & Configuration of types in the population & \sect{modeling_assumptions} \\
$x_{i}$ & Type occupying $i$ ($1$ for $A$, $0$ for $B$) & \sect{modeling_assumptions} \\
$\bm{\xi}$ & Initial configuration of types in the population & \sect{quantifying_selection} \\
\hline
${}^{\circ}$ & Indicates the absence of selection & \sect{modeling_assumptions} \\
$\widehat{}$ & Indicates weighting by reproductive values & \sect{quantifying_selection}
\end{longtable}

\end{document}